\documentclass[prd,aps,amsfonts,eqsecnum,superscriptaddress,nofootinbib,notitlepage,longbibliography,final]{revtex4-1}

\usepackage[english]{babel}
\usepackage[utf8]{inputenc}

\usepackage[a4paper, left=2cm, right=2cm, top=2cm, bottom=2cm]{geometry}

\usepackage{amsmath,physics,amsfonts}
\usepackage{graphicx,subcaption}
\usepackage[colorlinks=true, allcolors=blue]{hyperref}
\usepackage{float}
\usepackage{hhline}
\usepackage{mathtools}
\mathtoolsset{showonlyrefs}
\usepackage{soul}
\usepackage{tcolorbox}

\usepackage{pgfplots}
\pgfplotsset{compat=newest}

\usepgfplotslibrary{groupplots}

\definecolor{zero}{RGB}{180, 180, 180}
\definecolor{one}{RGB}{206, 104, 104}
\definecolor{two}{RGB}{231, 195, 146}
\definecolor{three}{RGB}{231, 217, 146}
\definecolor{four}{RGB}{209, 223, 141}
\definecolor{five}{RGB}{116, 185, 116}
\definecolor{six}{RGB}{126, 147, 165}
\definecolor{seven}{RGB}{146, 136, 176}
\definecolor{eight}{RGB}{166, 124, 166}

\usepackage{amsthm}

\newtheorem{theorem}{Theorem}[section]
\newtheorem{corollary}{Corollary}[theorem]
\newtheorem{lemma}[theorem]{Lemma}
\newtheorem{prop}[theorem]{Proposition}
\newtheorem{assumption}[theorem]{Assumption}

\theoremstyle{definition}
\newtheorem{definition}[theorem]{Definition}
\newtheorem{remark}[theorem]{Remark}

\DeclareMathOperator{\id}{\mathbf{1}}
\newcommand{\e}{\mathrm{e}}

\renewcommand\labelenumi{(\roman{enumi})}
\renewcommand\theenumi\labelenumi

\newcommand{\oscnorm}[1]{{\left\vert\kern-0.25ex\left\vert\kern-0.25ex\left\vert #1 
    \right\vert\kern-0.25ex\right\vert\kern-0.25ex\right\vert}}

\begin{document}

\title{Rapid Mixing of Quantum Gibbs Samplers for Weakly-Interacting Quantum Systems}
\author{Štěpán Šmíd}
\email{s (dot) smid23 (at) imperial (dot) ac (dot) uk}
\affiliation{%
Department of Computing, Imperial College London, United Kingdom
}%
\author{Richard Meister}
\affiliation{%
Department of Computing, Imperial College London, United Kingdom
}%
\author{Mario Berta}
\affiliation{%
Institute for Quantum Information, RWTH Aachen University, Germany
}%
\affiliation{%
Department of Computing, Imperial College London, United Kingdom
}%
\author{Roberto Bondesan}
\affiliation{%
Department of Computing, Imperial College London, United Kingdom
}%

\date{\today}


\begin{abstract}
    Dissipative quantum algorithms for state preparation in many-body systems are increasingly recognised as promising candidates for achieving large quantum advantages in application-relevant tasks. Recent advances in algorithmic, detailed-balance Lindbladians enable the efficient simulation of open-system dynamics converging towards desired target states. However, the overall complexity of such schemes is governed by system-size dependent mixing times. In this work, we analyse algorithmic Lindbladians for Gibbs state preparation and prove that they exhibit rapid mixing, i.e., convergence in time poly-logarithmic in the system size. We first establish this for non-interacting spin systems, free fermions, and free bosons, and then show that these rapid mixing results are stable under perturbations, covering weakly interacting qudits and perturbed non-hopping fermions. Further, we adapt the techniques from separable qudits to the fermionic setting and prove rapid mixing of the strongly-interacting regime of the Fermi-Hubbard model, for which we also explicitly evaluate the guaranteed parameter regimes.
    Our results constitute the first efficient mixing bounds for non-commuting qudit models and bosonic systems at arbitrary temperatures. Compared to prior spectral-gap-based results for fermions, we achieve exponentially faster mixing, further featuring explicit constants on the maximal allowed interaction strength. This not only improves the overall polynomial runtime for quantum Gibbs state preparation, but also enhances robustness against noise. Our analysis relies on oscillator norm techniques from mathematical physics, where we introduce tailored variants adapted to specific Lindbladians\,---\,an innovation that we expect to significantly broaden the scope of these methods.
\end{abstract}

\maketitle


\section{Overview}

\paragraph*{Quantum Gibbs states.} The theoretical study of thermalisation of quantum many-body systems weakly coupled to a heat bath at a fixed temperature dates back to the work of Davies \cite{davies1974markovian} in the seventies.\footnote{We refer to \cite{Mozgunov2020completelypositive,PhysRevB.102.115109} for modern discussions originating in quantum information.} Since then, the quantum Markov semigroups ${e^{t\mathcal{L}}}_{t\geq 0}$ generated by Lindbladians $\mathcal{L}$ of Davies type have been extensively studied in the mathematical physics literature. It has been shown for several commuting quantum Hamiltonians of interest that these Lindbladians exhibit \textit{fast} thermalisation, occurring within polynomial time in the system size, see, e.g., \cite{Temme_2013,temme2015thermalize,Kastoryano:2016aa} and references therein. The most direct approach to bounding the \textit{mixing time} is to analyse the Lindbladian spectrum and establish a constant lower bound on its gap. There are also techniques which can\,---\,in some cases\,---\,avoid the need of estimating the spectral gap of the full generator by treating the coherent and dissipative part separately \cite{Fang_2025}. Nonetheless, a polynomial growth of the mixing time can actually be a severe overestimate, and some Lindbladians can even thermalise \textit{rapidly}, that is, in poly-logarithmic time. Such bounds are not obtainable from the spectrum of the Lindbladian, but rather require a finer understanding of the structure of the dynamics. A technique using the \textit{modified logarithmic Sobolev inequality} has been developed to resolve this \cite{Kastoryano_2013} and used to prove rapid mixing of some types of systems, such as commuting Hamiltonians \cite{capel2020modified, kochanowski2024rapid,bardet2023rapid}. Fast converging open system dynamics then naturally lends itself for the task of algorithmic Gibbs state preparation. However, in spite of the progress on understanding thermalisation times, all these rapid mixing results were obtained with so-called Davies generators that face a major drawback when it comes to applications: the non-local nature of the transitions they induce prevents the existence of efficient simulation protocols.\\


\paragraph*{Algorithmic quantum Gibbs samplers.} Recent seminal works on quantum algorithms for Gibbs state preparation have managed to address this drawback, showing that it is possible to bring together algorithmic efficiency with an exact notion of \textit{quantum detailed balance} with the Gibbs state of interest (albeit a different notion than that obeyed by Davies generators) \cite{chen2023efficient, ding2024efficient, gilyen2024glauber}. These algorithmic Lindbladians mimic the thermalisation of Davies generators, but also retain simulability by enforcing locality of the jumps they induce. Implementing the dynamics they generate hence represents the (fully) quantum counterpart of classical \textit{Markov chain Monte Carlo} algorithms. Understanding the end-to-end complexity of the Gibbs state preparation again requires bounding the mixing times, which is most readily available through spectral gap estimates. This was previously achieved in the regime of high temperatures \cite{rouze2024efficient}, for weakly-interacting fermionic systems \cite{smid2025polynomial,Tong_2025Fast}, random Hamiltonians \cite{ramkumar2024mixingtimequantumgibbs, basso2024random}, and other problem specific cases \cite{ding2024toriccode,chen2024constanttemp,rajakumar2024contantlocal}. These rigorous bounds on mixing times are important, since they rigorously show the efficiency of quantum algorithms for quantum many-body state preparation, which is one of the most promising territories for an end-to-end quantum advantage \cite{lin2025dissipativepreparationmanybodyquantum}. As was the case for Davies generators, it is also possible to show rapid mixing in some cases of these novel algorithmic Lindbladians. This was again achieved for Gibbs state preparation in the regime of high temperatures \cite{rouze2024optimalquantumalgorithmgibbs}. Rapid mixing was then also proven (for some restricted cases) of ground state preparation of weakly-interacting spin and fermionic Hamiltonians \cite{zhan2025rapidquantumgroundstate}.\\


\paragraph*{Main results on mixing times.} We build upon these rapid mixing results, together with the fast mixing results for weakly-interacting fermionic systems \cite{smid2025polynomial,Tong_2025Fast}, and prove rapid mixing of quantum Gibbs samplers at any temperature for large classes of weakly-interacting quantum systems with different particle statistics. Specifically, we consider qudit spin systems, fermionic systems, and bosonic systems, and commence by proving that their respective non-interacting versions mix rapidly, which is (informally) summarised in the following proposition.\footnote{For bosonic Hamiltonians, the result on rapid mixing is restricted to the initial state being the vacuum state $\rho = \ket{\mathbf{0}}\bra{\mathbf{0}}$.}

\begin{tcolorbox}[colback=yellow!15,colframe=yellow!15, boxrule=0pt, left=2pt, right=2pt, top=2pt, bottom=2pt]
\begin{prop}[Rapid mixing of non-interacting systems, informal]
    For non-interacting quantum systems, including separable qudit spin Hamiltonians
    \begin{align}
        \text{$H_s = \sum_{i=1}^n h_i$ with each $h_i$ supported strictly at the qudit $i$,}
    \end{align} free fermionic Hamiltonians
    \begin{align}
        \text{$H_f = \sum_{i,j=1}^n c_i^\dagger M_{ij} c_j$ with $c_i$'s obeying $\{c_i,c_j^\dagger\} = \delta_{ij}$,}
    \end{align}
    and free bosonic Hamiltonians
    \begin{align}
        \text{$H_b = \sum_{i,j=1}^n a_i^\dagger h_{ij} a_j$ with $a_i$'s obeying $[a_i,a_j^\dagger] = \delta_{ij}$,}
    \end{align}
    the algorithmic Lindbladian $\mathcal{L}$ for Gibbs state preparation mixes rapidly in logarithmic time, \begin{equation}
        t^{(\textup{type})}_\textup{mix}(\epsilon) \leq c_1^{(\textup{type})} \cdot \log \left( c_2^{(\textup{type})} \cdot \frac{n}{\epsilon} \right)\,,
    \end{equation} with $c_1$ and $c_2$ being size-independent constants specified later, where $c_1 \sim \frac{1}{\operatorname{gap}(\mathcal{L})}$, and $\epsilon$ specifying the accuracy in trace distance.
\end{prop}
\end{tcolorbox}

Next, for spin and fermionic systems, we consider the effect of perturbing the Hamiltonians by a (quasi-)local interaction term $\lambda \cdot V$ with a coupling strength $\lambda$. In the case of spin systems, we prove that the Lindbladian remains rapidly mixing as long as $\lambda$ is small enough, while in the fermionic case we obtain this result only when the original Hamiltonian did not include hopping between fermions. This is summarised in the following theorem.

\begin{tcolorbox}[colback=yellow!15,colframe=yellow!15, boxrule=0pt, left=2pt, right=2pt, top=2pt, bottom=2pt]
\begin{theorem}[Rapid mixing of perturbed systems, informal]
    For separable qudits under (quasi-local) perturbation, given by
    \begin{align}
        H_s(\lambda) = \sum_{i=1}^n h_i + \lambda \cdot V,
    \end{align}
    and for non-hopping free fermions under (quasi-local) perturbation, given by
    \begin{align}
        H_f(\lambda) = \sum_{i=1}^n \epsilon_i c_i^\dagger c_i + \lambda \cdot V,
    \end{align}
    there exist corresponding maximal interaction strengths $\lambda_\textup{max}$ such that for $|\lambda| \leq \lambda_\textup{max}$, the Lindbladian $\mathcal{L}$ remains rapidly mixing. This holds in any fixed dimension, for any geometry, and at any constant temperature. Further, the constant $\lambda_\textup{max}$ can be explicitly evaluated for any given system. 
\end{theorem}
\end{tcolorbox}

Finally, we also adapt the techniques used for weakly-interacting qudit systems to the fermionic case, and prove rapid mixing for the strongly-interacting regime of the Fermi-Hubbard model. This can be summarised in the following proposition:\enlargethispage{1cm}

\begin{tcolorbox}[colback=yellow!15,colframe=yellow!15, boxrule=0pt, left=2pt, right=2pt, top=0pt, bottom=0pt]
\begin{prop}[Rapid mixing of strongly-interacting Fermi-Hubbard model]
    For the Fermi-Hubbard model given by \begin{equation}
    H_\textup{FH} = -t\sum_{\langle i,j\rangle, \sigma} \left( c^\dagger _{i,\sigma} c_{j,\sigma} + c^\dagger_{j,\sigma}c_{i,\sigma} \right) + U\sum_{i=1}^n N_{i,\uparrow} N_{i,\downarrow}\,,\end{equation} where $\langle i,j\rangle$ represents nearest neighbours on the lattice, $\sigma \in \{\uparrow,\downarrow\}$ represents the spin of the fermions, and $N_{i,\sigma} = c^\dagger _{i,\sigma} c_{i,\sigma}$, there exists a maximal hopping strength $t_\textup{max}$ (which depends on the fixed interaction strength $U$ and temperature $\beta$) such that for $|t|\leq t_\textup{max}$, the corresponding Lindbladian $\mathcal{L}$ mixes rapidly. Further, the constant $t_\textup{max}$ can be explicitly evaluated and is presented on Figure \ref{fig: certified regimes for FH}.
\end{prop}
\end{tcolorbox}


\paragraph*{Algorithmic corollaries.} Rapid mixing times have immediate consequences for the complexities of the corresponding quantum Gibbs samplers. Following the construction of \cite{ding2024efficient}, a logarithmic mixing time brings down the algorithmic complexity to \begin{align}
\text{$\widetilde{\mathcal{O}}(n^2 \operatorname{polylog}(1/\epsilon))$ Hamiltonian simulation time with $\mathcal{O}(n)$ qubits for Gibbs state preparation.}
\end{align} Efficient thermal state preparation can be then readily applied to the calculation of the \textit{partition function} for the corresponding systems, as explained in \cite{rouze2024optimalquantumalgorithmgibbs,smid2025polynomial}, yielding
\begin{align}
\text{$\widetilde{\mathcal{O}}(n^4 /\epsilon^2)$ Hamiltonian simulation time with $\mathcal{O}(n)$ qubits for the estimation of free energies.}
\end{align}
As such\,---\,and generalising our previous work on the atomic limit case \cite{smid2025polynomial}\,---\,our results provide the first provably efficient quantum algorithm for preparing thermal states of weakly-interacting qudit systems at any temperature. This includes for example spin Hamiltonians like the Heisenberg model in a strong external magnetic field. To the best of our knowledge, we also provide the first rapid mixing result for any bosonic Lindbladian; the difficulty of which stems from the non-boundedness of the bosonic ladder operators appearing therein. The results for fermionic systems are then applicable for example to systems where the chemical potential is the dominant term.\\


\paragraph*{Proof ideas.} As rapid mixing is not accessible by studying the spectrum of the generator on its own, and the modified log-Sobolev inequality has not yet been established in any case of these algorithmic Lindbladians, our techniques rely primarily on the study of the \textit{oscillator norm}, which was originally proposed in \cite{Majewski1995,Richter_1996}, and more recently used in \cite{temme2015thermalize,rouze2024optimalquantumalgorithmgibbs,zhan2025rapidquantumgroundstate}. The main technical novelty of our work lies in the generalisation and adaptation of the oscillator norm to each particular Lindbladian we consider, i.e.~making it problem-specific, and hence greatly extending its original applicability. We expect these ideas to pave the way for extensive follow-up work on understanding rapid mixing of many Lindbladians of interest. This technique further allows for an explicit evaluation of the maximal covered interaction strengths, resolving the vagueness of methods based on topological stability of the spectral gap \cite{rouze2024efficient,smid2025polynomial,Tong_2025Fast}.\\


\paragraph*{Impact of rapid mixing.} Faster mixing times also lead to shallower quantum circuits, making them more suitable for near term quantum devices. On top of that, it has been shown that the dynamics of rapidly mixing Lindbladians is stable under perturbation \cite{Cubitt_2015}, further improving their potential resilience to noise. Last but not least, the study of rapid mixing is of general interest for several other reasons, like implications on exponential decay of correlations in the steady states \cite{Kastoryano_2013_correl}, and applications to self-correcting quantum memories \cite{bergamaschi2025rapidmixinggibbsstates}.\\


\paragraph*{Conclusion.} Having at hand the demonstrated end-to-end polynomial time quantum algorithms for resolving the phase diagrams, promises strong quantum advantages in quantum simulation compared to state-of-the-art classical methods (although larger scale numerics for classically hard regimes would reveal more on that point). We emphasise that in contrast to popular quantum phase estimation based methods around ground state energy estimation\,---\,which suffer from the Ansatz state bottleneck\,---\,the presented complexities here are truly end-to-end. We believe our work constitutes a paradigm shift in the rigorous complexity analysis of algorithms for quantum simulation: It significantly advances the state of the art by showing that recent digital schemes based on Lindbladian dynamics can be made efficient for a wide class of physically relevant and classically hard models, offering pathways towards practical quantum advantage.\\


\paragraph*{Manuscript.} The rest of the paper is divided as follows: In Section \ref{sec:Background}, we briefly recapitulate the background for quantum Gibbs state preparation and rapid mixing. In the following, we present our novel results on the mixing times, covering qudits (Section \ref{sec:Spins}), fermions (Section \ref{sec:fermions}), and bosons (Section \ref{sec:bosons}). Finally, we provide an outlook (Section \ref{sec:Outlook}), where we discuss in particular some further thoughts on rapid mixing of general weakly-interacting fermions.


\section{Background}\label{sec:Background}

\subsection{Algorithmic quantum Gibbs samplers}

In this section, we shall briefly summarise the necessary information about the detailed-balanced Lindbladians, which serve as the generators for the algorithmic Gibbs state preparation, and whose mixing times we aim to bound. This summary is mostly taken from the full technical version of our previous work \cite{smid2025polynomial}. For a more in-depth exposition to these algorithms, we refer interested readers to the frameworks of \cite{chen2023efficient,ding2024efficient}, as well as the introductory review article \cite{lin2025dissipativepreparationmanybodyquantum}.\\

Quantum Gibbs sampling is the task of preparing the thermal state $\sigma_\beta = e^{-\beta H}/Z$ for a quantum Hamiltonian $H$. We first define the quantum Markov semigroup $\mathcal{P}_t$ as the semigroup of completely positive unital maps, and work with the generator of the dynamics called the Lindbladian $\mathcal{L}$: $\mathcal{P}_t = e^{t\mathcal{L}}$. Using Kraus' theorem, such a generator can be characterised by the following form
\begin{equation}\label{eqn:Lindbladian}
    \mathcal{L} [O] = i[G,O] + \sum_{a \in \mathcal{A}} \left( L_a^\dagger O L_a - \frac{1}{2} \{ L_a ^\dagger L_a, O \} \right) \,.
\end{equation} 
Here $L_a$ will be referred to as the Lindblad operators and $G=G^\dagger$ as the coherent term.

To perform quantum Gibbs sampling, we construct a quantum Markov semigroup such that $\lim_{t\to\infty}
\mathcal{P}_t^\dagger(\rho_0) = \sigma_\beta$
where $\rho_0$ is an arbitrary initial state and $\Phi^\dagger$ for a superoperator $\Phi$ is the adjoint w.r.t.~the Hilbert-Schmidt inner product: $\langle A,B\rangle=\Tr(A^\dagger B)$.
To define quantum detailed balance, we shall use the Kubo-Martin-Schwinger (KMS) inner product.
Given a full rank state $\sigma>0$, this is defined for two operators $A,B$ as
$\langle A,B\rangle_{\sigma}
    =
    \Tr(
    A^\dagger
    \mathcal{G}_\sigma
    (B))
    \,,$
where 
$
\mathcal{G}_\sigma(A) = 
\sigma^{1/2}
A
\sigma^{1/2}    
\,.
$

\begin{definition}(Quantum Detailed Balance)
    A Lindbladian $\mathcal{L}$ satisfies the KMS quantum detailed balance (QDB) condition if $\mathcal{L}$ is self-adjoint with respect to the KMS inner product.
\end{definition}

Since $\mathcal{L}[\id]=0$, we have that if $\mathcal{L}$ satisfies QDB,
\begin{align}
    0
    =
    \langle A, \mathcal{L}[\id]\rangle_\sigma
    =
    \langle \mathcal{L}[A], \id\rangle_\sigma
    =
    \langle \mathcal{L}[A], \sigma\rangle
    =
    \langle A, \mathcal{L}^\dagger[\sigma]\rangle
    \,,
\end{align}
for any operator $A$.
This shows that $\mathcal{L}^\dagger[\sigma]=0$
so that $\sigma$ is a stationary state of the dynamics generated by $\mathcal{L}^\dagger$. We can write the QDB condition more explicitly as:
\begin{equation}
\label{eq:qdb}
    \mathcal{L}
    = 
    \mathcal{G}_\sigma^{-1}
    \circ
    \mathcal{L}^\dagger 
    \circ
    \mathcal{G}_\sigma \,.
\end{equation}
Note that in general $\mathcal{L}$ is a non-Hermitian operator, however the self-adjointness with the KMS inner product guarantees real spectrum (which will be discussed closer in Lemma \ref{lemma:properties Falpha}).

The efficiency of the Lindbladian dynamics to prepare a thermal state is governed by its mixing time:
\begin{definition}
    The mixing time of the Lindbladian $\mathcal{L}^\dagger$ is \begin{equation}
        t_\textup{mix}(\epsilon) = \inf \left\{ t \geq 0 \left|\, \forall \rho: \left\|e^{t\mathcal{L}^\dagger}[\rho] - \sigma_\beta\right\|_{\Tr} \leq \epsilon \right.\right\}\,,
    \end{equation} where $\|A\|_{\Tr}
    =\Tr(\sqrt{A^\dagger A})
    $ denotes the trace norm. 
\end{definition}

\begin{definition}
    A mixing time scaling polynomially in the system size $n$ will be referred to as \textit{fast}, which can be equivalently characterised by the following contractivity of the Lindbladian:
\begin{align}
    \|e^{t\mathcal{L}^\dagger}[\rho]-\sigma\|_{\Tr} \leq \exp(\operatorname{poly}(n)) e^{- t/\operatorname{poly}(n)}\,.
\end{align}
A mixing time scaling poly-logarithmically will be referred to as \textit{rapid}, which can be equivalently characterised by
\begin{align}
    \|e^{t\mathcal{L}^\dagger}[\rho]-\sigma\|_{\Tr} \leq \operatorname{poly}(n) e^{-\alpha t}
\end{align}
for some constant $\alpha >0$.
\end{definition}

Next, we review the construction of a Lindbladian that satisfies QDB for $\sigma=\sigma_\beta$.
We follow the construction of \cite{ding2024efficient}\,---\,the main difference from the construction of \cite{chen2023efficient} is that it allows one to use a finite number of Lindblad operators. The construction is given in terms of a set of self-adjoint operators $\{A_a\}_{a\in \mathcal{A}}$ called jump operators and filter functions 
$\{\hat{f}^a(\nu)\}_{a\in \mathcal{A}}$
obeying
\begin{equation}
    \label{eq:q_fcn}
    \hat{f}^a(\nu) = q^a (\nu) e^{-\beta \nu /4},\quad q^a(-\nu) = \overline{q^a(\nu)} 
    \,.
\end{equation}
A popular filter function that we will mostly focus on below is the Gaussian one
\begin{align}
\label{eq:Gaussian filter}
    \hat f(\nu) = e^{-(\beta\nu+1)^2/8 + 1/8}\,,\quad 
    f(t) = \frac{1}{2\pi} \int_{-\infty}^{\infty}\hat{f}(\nu)e^{-i\nu t}\ \dd \nu
    =
    \sqrt{\frac{2}{\pi\beta^2}}
    \exp\left(-\tfrac{2}{\beta^2}
    \left(t-i\tfrac{\beta}{4}\right)^2\right)
    \,,
\end{align}
so that $f(t+i\beta/4)$ is real and positive. 
The Lindblad operators are then given by the filtered operator Fourier transforms of the jump operators \begin{align}
    L_a = \hat{f} ^a (\operatorname{ad}_H) A^a
    = \sum\limits_{\nu \in B_H} \hat{f}^a (\nu) A^a _\nu
    = 
    \int_{-\infty} ^{\infty} f^a (t)  e^{iHt} A^a e^{-iHt}\ \dd t \,,
\end{align}
where $B_H=\{\nu=E_i-E_j\ |\ E_i \in \operatorname{spec}(H)\}$ is the set of Bohr frequencies, and 
\begin{align}
    A_\nu = \sum_{i,j|E_i-E_j=\nu}
    P_iAP_j
    \,,\quad 
    A=\sum_{\nu\in B_H}A_\nu
    \,,\quad 
    A_\nu^\dagger = A_{-\nu}
    \,,
\end{align}
with $P_i$ the projector onto the eigenspace of eigenvalue $E_i$. Here, $\operatorname{ad}_H X = [H,X]$ represents the adjoint endomorphism of $H$. Note that $\operatorname{ad}_H A_\nu = [H,A_\nu]=\nu A_{\nu}$.
Further, $\hat{f}^a(\nu)$
denotes the Fourier transform of $f^a(t)$.
The coherent term is given by \begin{align}\hspace{-0.4cm}
    G &= -i \tanh \circ \log (\Delta ^{1/4} _{\sigma_\beta}) \left( \frac{1}{2} \sum_{a \in \mathcal{A}} L_a ^\dagger L_a \right)
    = \frac{i}{2} \sum\limits_{a \in \mathcal{A}} \sum\limits_{\nu \in B_H} \tanh\left( \frac{\beta \nu}{4} \right) (L_a ^\dagger L_a)_\nu
    = \sum_{a \in \mathcal{A}} \int_{-\infty}^\infty g(t) e^{iHt} (L_a ^\dagger L_a) e^{-iHt}\ \dd t \hspace{-0.5cm}\label{eqn:Coherent term}
\end{align} 
with
$
    \hat{g} (\nu) = \frac{i}{2} \tanh\left( \frac{\beta \nu}{4} \right) \cdot \kappa(\nu)\,,
$ where $\Delta_\rho [X] = \rho X \rho^{-1}$ is the modular superoperator, and $\kappa(\nu)$ is a sort of smooth indicator function, obeying $\kappa(\nu) = 1$ on $\nu \in [-2\|H\|,2\|H\|]$, and decaying smoothly and rapidly afterwards, so that it belongs to the class of Gevrey functions as per \cite[Equation (3.17)]{ding2024efficient}.
Here and in the following, $\|\cdot \|$ denotes the operator norm.
In \cite{ding2024efficient}, it was proven that the Lindbladian so defined satisfies KMS-QDB with the thermal state $\sigma_\beta$.\\

Reference \cite[Theorem 34]{ding2024efficient} also proves that this Lindbladian evolution can be simulated on a quantum computer 
up to time $t$ 
with Hamiltonian simulation time complexity
\begin{equation}
\label{eq:runtime_quantum}
\tilde{\mathcal{O}}(t (\beta + 1) |\mathcal{A}|^2 \log^{1+s}(1/\epsilon))\,,
\end{equation} 
where now $\epsilon$ is the precision of the channel in the diamond norm, and $s \geq 1$ is the Gevrey order of the filter function $\hat{f}(\nu)$ (which is for example equal to $1$ for the Gaussian filter).
This assumes normalisation of the jump operators of the form $\max_{a \in \mathcal{A}} \|A^a\| \leq 1$, access to their block encodings, access to controlled Hamiltonian simulation, and preparation oracles for the filter function $f(t)$ (where $f^a(t) = f(t)$ is taken to be the same for all $a \in \mathcal{A}$) and coherent function $g(t)$.


\subsection{Rapid mixing and oscillator norm}

To show fast mixing of a Lindbladian, the most direct approach is to study its eigenspectrum. More specifically, by proving that the spectral gap\footnote{Between the highest and second highest eigenvalue.} of the Lindbladian decays at most polynomially in the system size, we can show fast mixing using Hölder's inequality as \begin{align}
    \left\| e^{\mathcal{L}^\dagger t} [\rho] - \sigma_\beta \right\|_{\Tr} \leq e^{-\Delta(\mathcal{L}^\dagger)t} \left\| \sigma_\beta^{-1/2} \right\| \|\rho-\sigma_\beta\|_{\Tr}.
\end{align} A key intricacy when trying to show rapid mixing is that the knowledge of the spectrum of the Lindbladian on its own does not suffice. Instead, it requires a finer understanding of the dynamics of the quantum Markov semigroup. One such approach is based on the modified (or quantum) logarithmic Sobolev inequality \cite{Kastoryano_2013}. This inequality can be equivalently stated as the decay rate of the relative entropy of the evolved state with respect to the steady state. However, this inequality has been established only in a very limited number of cases, none of which included the algorithmic Lindbladians we consider here for Gibbs sampling. Instead, we will focus here on a so-called \textit{oscillator norm}, originally defined in \cite{Richter_1996,Majewski1995}.

\begin{definition}\label{def:original oscillator norm}
    Let $B(H)$ be the set of bounded operators over a Hilbert space $H$. Assume we have a set of bounded linear maps $\delta_I$ on $B(H)$ called \textit{quasi-derivations}, such that $\delta_I(\id)=0$, the set $\{O\in B(H): \sum_I \|\delta_I(O)\| < \infty\}$ is dense in $B(H)$ under spectral norm, and $\|O-C(O)\| \leq \sum_I \|\delta_I(O)\|$ where $C(O)$ is the normalised trace of $O$. Then \begin{align}
        \oscnorm{O} = \sum_I \|\delta_I(O)\|
    \end{align} defines a seminorm on $B(H)$ called the \textit{oscillator norm}. Specifically, for qubit spin systems, we take the index set to be the individual qubits and take the quasi-derivations to be $\delta_i(O) = O - \frac{1}{2}\Tr_i(O)$ for each qubit $i$.
\end{definition}

Note that the qubit-specific definition of quasi-derivations obeys the general condition $\|O-C(O)\| \leq \sum_I \|\delta_I(O)\|$ because of the following:
Let $C_I(A)=2^{-|I|}\Tr_{I}(A)$ denote the normalised trace on subsystem $I$. Note that $C_I \circ C_J = C_{I\cap J}$. Then we have the following telescoping sum: \begin{align}
        O - C(O)
        &= 
        O 
        - 
        C_1(O)
        +
        \sum_{i=2}^{n} 
        C_{[i-1]}
        \left( 
        O
        -
        C_{i}(O)
        \right)
        =
        \delta_1(O)
        +
        \sum_{i=2}^{n}
        C_{[i-1]} \delta_i(O)
        \,.
    \end{align} 
By taking the spectral norm, we get that    
\begin{align}
        \|O - C(O)\|
        \le 
        \sum_{i=1}^n 
        \| C_{[i-1]}\delta_i(O)\|
        \le 
        \sum_{i=1}^n 
        \|\delta_i(O)\|\,,
    \end{align}
where the last inequality follows since $C_I$ is a contraction. Later, we shall adapt this qubit-specific definition of the quasi-derivations to the case of qudits of dimension $d$ and fermions in Sections \ref{sec:Spins} and \ref{sec:fermions} respectively.

We can directly relate the oscillator norm to the mixing time of the Lindbladian as follows:
Denoting $\mathcal{P}^t=\e^{t\mathcal{L}}$, we can write the contractivity of the Lindbladian like
    \begin{align}
        \|\rho(t)-\sigma\|_{\Tr} 
        &= \sup_{\|O\|\leq 1} \Tr(O(\mathcal{P}^t)^\dagger(\rho-\sigma)) 
        = \sup_{\|O\|\leq 1} \Tr(\mathcal{P}^tO(\rho-\sigma)) 
        = \sup_{\|O\|\leq 1} \|O(t) - \Tr(O(t))/2^n\| \cdot \|\rho - \sigma\|_{\Tr}\,\\
        &\le 
        2
        \sup_{\|O\|\leq 1}
        \sum_{i=1}^n 
        \|\delta_i(O(t))\| = 2
        \sup_{\|O\|\leq 1} \oscnorm{O(t)}\,.
    \end{align}
Hence, since we can trivially bound $\oscnorm{O} \leq 2n$ for any $O$ obeying $\|O\| \leq 1$, it suffices to show that $\oscnorm{O(t)} \leq e^{-\alpha t} \oscnorm{O(0)}$ for some constant $\alpha > 0$ to prove rapid mixing of the Lindbladian $\mathcal{L}$. This has been previously done for the algorithmic Lindbladian for Gibbs state preparation in the regime of high temperatures \cite{rouze2024optimalquantumalgorithmgibbs}, as well as for ground state preparation of some weakly-interacting systems \cite{zhan2025rapidquantumgroundstate}. The main idea of both of these works is to firstly establish this decay for a simple, base-case Lindbladian (the unperturbed one), and then bound the effects of a weak enough perturbation stemming from locality of the Lindbladian.

The overarching challenge we will encounter in this work is that the unperturbed Lindbladians will be more complicated than the ones in \cite{rouze2024optimalquantumalgorithmgibbs,zhan2025rapidquantumgroundstate}, and so the standard oscillator technique won't be immediately applicable. Instead, our key technical contribution will be tailoring the oscillator norm for each Lindbladian, making it problem-specific, and hence greatly extending its usability. The oscillator norm generalised in this way will be an upper bound for the original one from Definition \ref{def:original oscillator norm}, so that it will still be usable for upper bounding the mixing time, but it will also allow the necessary step for bounding the decay of $\oscnorm{O(t)}$ to go through (see Sections \ref{sec: rapid mixing of qudits} and \ref{sec: rapid mixing of free fermions} later).


\section{Weakly-Interacting Quantum Spin Systems}\label{sec:Spins}

\subsection{Introduction}

Non-interacting quantum spin systems are described by separable Hamiltonians, $H^0=\sum_{i=1}^n h_i$, where $h_i$ has support only on the $i$-th site, which we consider to be a qudit of dimension $d$. In this section, we show that the Lindbladians introduced above mix rapidly, i.e.~in time logarithmic in $n$, for weakly-interacting quantum systems whose Hamiltonian is of the form $H=H^0+\lambda V$ with $|\lambda|$ smaller than a certain $n$-independent value, which can be evaluated explicitly for any given system.
Weakly interacting quantum spin models are ubiquitous in condensed matter physics. Paradigmatic examples are qubit lattice Hamiltonians of the form
\begin{align}
    H = \sum_{i=1}^n \alpha_i Z_i
    + \sum_{\langle i, j\rangle }
    h_{ij}
\end{align}
where $h_{ij}$ is a spin-spin interaction, such as the Ising interaction, $h_{ij}= J_{ij} X_i X_j$, or the Heisenberg interaction, $h_{ij}=J_{ij} (X_i X_j+Y_i Y_j+Z_i Z_j)$, where $J_{ij}=\mathcal{O}(\lambda)$.\\

Weakly interacting quantum systems have been studied in several previous works, which established stability of their spectrum and properties under small-enough couplings \cite{datta1996low,Yarotsky_2005,Bravyi_2011}.
Classically efficient algorithms exist for their ground states  \cite{Bravyi_2008} and thermal states at low temperatures \cite{Helmuth_2023}. We note, however, that in \cite{Helmuth_2023} the time complexity of the algorithm is a polynomial whose degree depends on the graph connectivity, and no explicit formula is given, so that it is not possible to directly compare our quantum Gibbs sampling runtime with that of this classical algorithm. We further note that \cite{zhan2025rapidquantumgroundstate} studies the mixing time of a restricted class of weakly interacting quantum systems, however for a Lindbladian that prepares the ground state only.\\

Here, we derive a rigorous result on the mixing time for preparing the Gibbs state of weakly-interacting, non-commuting, quantum spin systems of qudits at any temperature. Combined with the efficiency of the quantum Gibbs sampling algorithm of \cite{chen2023efficient,ding2024efficient}, our result leads to the \emph{first provably efficient quantum algorithm for preparing the Gibbs state of weakly interacting spin systems}, whose 
Hamiltonian simulation time complexity is $\tilde{\mathcal{O}}(n^2)$. On a technical level, our result can be seen as an extension of  
\cite{rouze2024optimalquantumalgorithmgibbs},
which considers perturbations of the depolarising channel for qubits, corresponding to the Lindbladian at infinite temperature. 
Our proof builds on \cite{rouze2024optimalquantumalgorithmgibbs} and introduces an adapted oscillator norm technique that allows us to deal with qudit systems at any temperature.


\subsection{Notation and some lemmas}


We first derive some general properties of the eigenvectors of $\mathcal{L}$ that we will use later in the proof of rapid mixing, specifically for defining and using the Lindbladian-specific version of the oscillator norm.
In the following,  $[N]$ denotes the set $\{1,\dots,N\}$.

\begin{lemma}
    \label{lemma:properties Falpha}
    Let $\mathcal{L}$ be a Lindbladian acting on a quantum system of dimension $N$.
    Assume that $\mathcal{L}$ is irreducible and satisfies the KMS quantum detailed balance condition for a state $\sigma$.
    Then:
    \begin{enumerate}
    \item 
    $\mathcal{L}$ has eigenvectors $\{ F^\alpha \}$ 
    that are orthonormal w.r.t.~the KMS inner product.
    The eigenvalues are $0=\lambda^1 > \lambda^2 \ge \cdots \ge \lambda^{N^2}$.
    \item 
    For any linear operator $F$, we have
    \begin{align}
        \mathcal{L}F
        =\sum_{\alpha=1}^{N^2}
        \lambda^\alpha
        F^\alpha 
        \langle F^\alpha, F\rangle_\sigma
    \end{align}
    \item 
    With $\sigma=\e^{-\beta H}/Z$, let $\Delta E=E_{\textup{max}}-E_{\textup{min}}$ be the difference between the largest and smallest eigenvalue of $H$. 
    Then for any linear operator $F$ and $\alpha\in [N^2]$:
    \begin{align}
        \| F^\alpha 
        \langle F^\alpha, F\rangle_\sigma
        \|
        &\le 
        N^{2}
        \e^{2\beta \Delta E}
        \|F\|
    \end{align}
    \end{enumerate}    
\end{lemma}

\begin{proof}
    \begin{enumerate}
        \item 
        Introduce the self-adjoint parent Hamiltonian:
        \begin{align}
            \mathcal{H}
            =
            \mathcal{G}_\sigma^{+1/2}
            \mathcal{L}
            \mathcal{G}_\sigma^{-1/2}\,.
        \end{align}
        $\mathcal{H}$ has eigenvectors $V^\alpha$ orthonormal w.r.t.~the Hilbert-Schmidt inner product 
        and has the same spectrum as $\mathcal{L}$, where the properties of the spectrum follow from the assumptions  on $\mathcal{L}$.
        $F^\alpha= \mathcal{G}_\sigma^{-1/2}V^\alpha$ is an eigenstate of $\mathcal{L}$ with eigenvalue $\lambda^\alpha$:
        \begin{align}
            \mathcal{L}F^\alpha
            =
            \mathcal{G}_\sigma^{-1/2}
            \mathcal{H}
            \mathcal{G}_\sigma^{+1/2}
            \mathcal{G}_\sigma^{-1/2}V^\alpha
            =
            \lambda^\alpha F^\alpha
            \,,
        \end{align}
        and
        \begin{align}
            \delta_{\alpha,\alpha'}
            =
            \langle V^\alpha,V^{\alpha'}
            \rangle 
            =
            \langle \mathcal{G}_\sigma^{+1/2}
            F^\alpha,
            \mathcal{G}_\sigma^{+1/2}F^{\alpha'}
            \rangle 
            =
            \langle 
            F^\alpha,
            \mathcal{G}_\sigma F^{\alpha'}
            \rangle 
            =
            \langle 
            F^\alpha, F^{\alpha'} 
            \rangle_\sigma
        \end{align}
        Note that $V^1 = \sqrt{\sigma}$ since $F^1=\id$.
        \item 
        Follows from the orthonormality of the basis 
        $\{ F^\alpha \}$ and the detailed balance condition:
        \begin{align}
            \mathcal{L}F
            =
            \sum_\alpha 
            F^\alpha 
            \langle 
            F^\alpha, 
            \mathcal{L}F
            \rangle_\sigma 
            =
            \sum_\alpha 
            F^\alpha 
            \langle 
            \mathcal{L}
            F^\alpha, 
            F
            \rangle_\sigma 
            =
            \sum_\alpha 
            \lambda^\alpha 
            F^\alpha 
            \langle 
            F^\alpha, 
            F
            \rangle_\sigma 
        \end{align}
    \item 
    Consider first
    $\| F^\alpha \|$:
    \begin{align}
        \|F^\alpha\|
        =
        \| 
        \sigma^{-1/4}
        V^\alpha
        \sigma^{-1/4}
        \|
        \le 
        \|\sigma^{-1/4}\|^2 
        \|
        V^\alpha \|
        \le 
        \|\sigma^{-1/2}\|
    \end{align}
    The first inequality follows from the submultiplicativity of the operator norm and the second from bounding the operator norm of $V^\alpha$ by the Hilbert-Schmidt norm 
    $\langle V^\alpha, V^\alpha\rangle = 1$.
    Next, we have explicitly:
    \begin{align}
        \|F^\alpha\|
        \le 
        \|\sigma^{-1/2}\|
        =
        Z^{1/2}
        \|\e^{+\frac{\beta}{2} H}\|
        =
        Z^{1/2}
        \e^{+\frac{\beta}{2} E_{\text{max}}}
        \le 
        N^{1/2} 
        \e^{-\frac{\beta}{2} E_{\text{min}}}
        \e^{+\frac{\beta}{2} E_{\text{max}}}
        =
        N^{1/2} 
        \e^{+\frac{\beta}{2} \Delta E}
        \,.
    \end{align}
    Next we bound $|\langle F^\alpha, F\rangle_\sigma|
    =
    |\Tr((F^\alpha)^\dagger 
    \sigma^{1/2}
    F 
    \sigma^{1/2}
    )|
    $:
    \begin{align}
        |\langle F^\alpha, F\rangle_\sigma|
        \le 
        \|F^\alpha \|
        \|
        \sigma^{1/2}
        F
        \sigma^{1/2}\|_{\Tr}
        \le 
        \|F^\alpha \|
        \|
        \sigma^{1/2}\|_{\Tr}
        \,
        \|
        F
        \sigma^{1/2}\|_{\Tr}
        \le 
        \|F^\alpha \|
        \|F \|
        \|
        \sigma^{1/2}\|_{\Tr}^2      \,.
    \end{align}
    The first inequality uses the tracial matrix Hölder inequality, the second the submultiplicativity of the trace norm, the third the 
    Hölder inequality for Schatten norms.
    Finally, we have explicitly,
    denoting the eigenvalues of $H$ by $E_\ell$:
    \begin{align}
        \|
        \sigma^{1/2}\|_{\Tr}^2
        =
        \frac{
        \left(
        \sum_\ell \e^{-\frac{\beta}{2}E_\ell}\right)^2}
        {\sum_\ell \e^{-\beta E_\ell}}
        \le 
        \frac{N^2 \e^{-\beta E_{\text{min}}}}{N 
        \e^{-\beta E_{\text{max}}}
        }
        =
        N \e^{\beta \Delta E}
    \end{align}
    Putting things together we get the result of the Lemma.
    \end{enumerate}
\end{proof}

A key ingredient for bounding the mixing times of the algorithmic Lindbladian of \cite{ding2024efficient}
specified in Equation \eqref{eqn:Lindbladian} is understanding its locality and characterising the strength of the perturbation for weakly interacting systems. We show these properties in detail in the Appendix in Lemma \ref{lemma:bounds on delta L}. These are the counterpart of the analysis for high temperatures in \cite[App.~B]{rouze2024optimalquantumalgorithmgibbs}. To complete the proof of rapid mixing, we will also need the following Lemma \ref{lemma:eta} providing us the constants appearing in the upper bound for the mixing time after perturbing the Lindbladian, which is an adaptation of \cite[App.~A 2]{rouze2024optimalquantumalgorithmgibbs} and the proof follows the same steps. Its assumptions would then be satisfied using the results from Lemma \ref{lemma:bounds on delta L}.

\begin{lemma}\label{lemma:eta}
    Let $\mathcal{L}$ and $\mathcal{L}^0$ be two Lindbladians that satisfy
    \begin{align}
        &\| \mathcal{L}^{(r)}_i - \mathcal{L}^{(r-1)}_i\|_{\infty\to\infty}\le \zeta(r)
        \,,\quad 
        \| \mathcal{L}_i - \mathcal{L}^0_i\|_{\infty\to\infty}\le \xi(\lambda)
    \end{align}
    for some positive functions $\zeta(r),\xi(\lambda)$ such that 
    $\Delta(\ell) = 
    \sum_{r\ge \ell}\zeta(r)$ is finite.
    Here, $\mathcal{L}^{(r)}_i$ is the truncation of the Lindbladian $\mathcal{L}_i$ associated to the Hamiltonian $H$ truncated to a ball of radius $r$ around site $i$.
    Let $\{ \delta_i^\alpha \}$ be a set of partial quasi-derivations such that $\sum_{\alpha \in S}\delta_i^\alpha(O)=\delta_i(O)=O-\frac{1}{d}\Tr_i(O)$ is the full quasi-derivation at the qudit $i$, with $\|\delta_i^\alpha O\|\le d_\star \|O\|$, where $d$ is the dimension of the qudits.
    
    Then, for any positive $r_0$, linear operator $O$ with $\|O\|\le 1$, $\alpha\in S$ and $i\neq j$, we have 
    \begin{align}
        \label{eq:relations kappa gamma}
        &\|[
        \delta_i^\alpha,
        \mathcal{L}_j
        ](O)\|
        \le 
        \sum_{\alpha'}
        \sum_\ell
        \kappa_{ij}^{\ell,\alpha'}
        \| \delta_\ell^{\alpha'} O\|
        \,,\quad 
        \|\delta_i^\alpha 
        (
        \mathcal{L}_i
        -
        \mathcal{L}_i^0) O\|
        \le 
        \sum_{\alpha'}
        \sum_\ell \gamma_i^{\ell,\alpha'}
        \| \delta_\ell^{\alpha'} O\|
        \,,
    \end{align}
    for some $\kappa_{ij}^{\ell,\alpha}, \gamma_i^{\ell,\alpha}\ge 0$ such that 
    \begin{align}
        \sum_{i}
        \Big(
        \gamma_i^{\ell,\alpha}
        +
        \sum_{j\neq i}
        \kappa_{ij}^{\ell,\alpha}
        \Big)
        \le 
\xi(\lambda)(1+2d_\star)(2r_0+1)^{2D}
+f(r_0)
\equiv \eta 
    \end{align}    
    where
    \begin{align}
        f(r_0)&=
\Delta(r_0)(1+2d_\star)(2r_0+1)^{2D}
+
\left(d_\star(2+r_0) + d_\star (2r_0+1)^D+1 \right)
\sum_{\ell>r_0}
        \Delta(\ell)
        (2\ell+1)^{2D}
\\&\quad
+ d_\star 
\sum_{\ell'\ge r_0}
\sum_{\ell \ge \ell'}
\Delta(\ell) (2\ell+1)^{2D}\label{eqn:f_decay function}
    \end{align}
\end{lemma}
\begin{proof}
We have, denoting the distance between $i,j$ by $d_{ij}$ and noting that $\mathcal{L}_j^{(d_{ij})}$ is not supported at $i$:
\begin{align}
    \| [\delta_i^\alpha,\mathcal{L}_j](O)\|
    =
    \| [\delta_i^\alpha,\mathcal{L}_j-\mathcal{L}_j^{(d_{ij})}](O)\|
    \le 
    \|
    \mathcal{L}_j-\mathcal{L}_j^{(d_{ij})}\|_{\infty\to\infty}
    \|\delta_i^\alpha(O)\|+
    d_\star 
    \|(\mathcal{L}_j-\mathcal{L}_j^{(d_{ij})})(O)\|\,.
\end{align}
Using the telescopic sum 
$\mathcal{L}_j-\mathcal{L}_j^{(r_0)}
=
\sum_{r>r_0}\mathcal{L}_j^{(r)}-\mathcal{L}_j^{(r-1)}$ and the fact that $\mathcal{L}_j^{(r-1)}(\id_{B_j(r)})=0$ and another telescopic telescopic sum 
$\delta_{B_j(r)}O
=\delta_1O+\sum_{i=2}^{|B_j(r)|}C_{[i-1]}\delta_iO$ with $C_I$ the normalised trace over $I$, we have
\begin{align}
    \| [\delta_i^\alpha,\mathcal{L}_j](O)\|
    &\le 
    \sum_{r>d_{ij}}
    \zeta(r)
    \|\delta_i^\alpha(O)\|+
    d_\star 
    \sum_{r>d_{ij}}
    \|(\mathcal{L}^{(r)}_j-\mathcal{L}_j^{(r-1)})
    \delta_{B_j(r)}(O)\|
    \\
    &\le
    \Delta(d_{ij})
    \|\delta_i^\alpha(O)\|+
    d_\star 
    \sum_{r>d_{ij}}
    \zeta(r)
    \sum_{k|d_{jk}\le r}
    \|\delta_{k}(O)\|
    \\
    &\le
    \Delta(d_{ij})
    \|\delta_i^\alpha(O)\|+
    d_\star 
    \sum_{\alpha'}
    \sum_k
    \sum_{r>\max(d_{ij},d_{jk})}
    \zeta(r)
    \|\delta_{k}^{\alpha'}(O)\|
    \\
    &=
    \Delta(d_{ij})
    \|\delta_i^\alpha(O)\|+
    d_\star 
    \sum_{\alpha'}
    \sum_k
    \Delta(\max(d_{ij},d_{jk}))
    \|\delta_{k}^{\alpha'}(O)\|\,.\label{eqn: kappas small}
\end{align}\enlargethispage{1cm}
Now, we fix $r_0>0$ to be chosen later and separate the cases $d_{ij}>r_0$, for which we use the above expression, and
$d_{ij}\le r_0$, where we use the following bound: Since $i\neq j$ and
$\mathcal{L}_j^0$ is supported strictly at $j$, we have
\begin{align}
    &\| [\delta_i^\alpha,\mathcal{L}_j](O)\|
    \le 
    \|[\delta_i^\alpha,
    \mathcal{L}_j-\mathcal{L}_j^{(r_0)}](O)
    \|
    +
    \|[\delta_i^\alpha,
    \mathcal{L}_j^{(r_0)}-\mathcal{L}_j^{0}](O)
    \|
    \\
    &\le 
    d_\star 
    \|
    (\mathcal{L}_j-\mathcal{L}_j^{(r_0)})(O)
    \|
    +
    \|
    (\mathcal{L}_j-\mathcal{L}_j^{(r_0)})
    \|_{\infty\to\infty}
    \|\delta_i^\alpha(O)\|
    +
    d_\star 
    \|
    (\mathcal{L}_j^{(r_0)}-\mathcal{L}_j^{0})(O)
    \|
    +
    \|
    (\mathcal{L}_j^{(r_0)}-\mathcal{L}_j^0)
    \|_{\infty\to\infty}
    \|\delta_i^\alpha(O)\|
\end{align}
Now note that the bound on the strength of the perturbation holds also for $\mathcal{L}_j^{(r_0)}-\mathcal{L}_j^0$, and proceeding as above by introducing the quasi-derivation $\delta_{B_j(r)}$:
\begin{align}
    &\| [\delta_i^\alpha,\mathcal{L}_j](O)\|
    \le 
    (\Delta(r_0)+\xi(\lambda))
    \|\delta_i^\alpha(O)\|
    +
    d_\star 
    \sum_{r>r_0}
    \|
    (\mathcal{L}_j^{r}-\mathcal{L}_j^{(r-1)})
    \delta_{B_j(r)}(O)
    \|
    +
    d_\star 
    \|
    (\mathcal{L}_j^{(r_0)}-\mathcal{L}_j^{0})\delta_{B_j(r_0)}(O)
    \|\\
    &\le 
    (\Delta(r_0)+\xi(\lambda))
    \|\delta_i^\alpha(O)\|
    +
    d_\star 
    \sum_{\alpha'}
    \left(
    \sum_k \Delta(\max(d_{jk},r_0))
    +
    \xi(\lambda)
    \sum_{k|d_{jk}\le r_0}
    \right)
    \|
    \delta_{k}^{\alpha'}(O)
    \|\\
    & 
    \le 
    (\Delta(r_0)+\xi(\lambda))
    \|\delta_i^\alpha(O)\|
    +
    d_\star 
    \sum_{\alpha'}
    \sum_{k|d_{jk}\le r_0}
    \left(
    \Delta(r_0)
    +
    \xi(\lambda)
    \right)
    \|
    \delta_{k}^{\alpha'}(O)
    \|
    +
    d_\star 
    \sum_{\alpha'}
    \sum_{k| d_{jk}>r_0}
    \Delta(d_{jk})
    \|
    \delta_{k}^{\alpha'}(O)
    \|\\
    &\le 
    (1+d_\star)
    \left(
    \Delta(r_0)
    +
    \xi(\lambda)
    \right)
    \sum_{\alpha'}
    \sum_{k|d_{jk}\le r_0}
    \|
    \delta_{k}^{\alpha'}(O)
    \|
    +
    d_\star 
    \sum_{\alpha'}
    \sum_{k| d_{jk}>r_0}
    \Delta(d_{jk})
    \|
    \delta_{k}^{\alpha'}(O)
    \|
    \,.\label{eqn: kappas large}
\end{align}
From Equations \eqref{eqn: kappas small} and \eqref{eqn: kappas large} we can read off the values of $\kappa$'s in the first desired inequality to be
\begin{align}
    \kappa_{ij}^{k,\alpha'}
    =
    \begin{cases}
       \Delta(d_{ij})
       & i=k, d_{ij}>r_0\,,\\
       d_\star 
    \Delta(\max(d_{ij},d_{jk})) & i\neq k, d_{ij}>r_0\,,\\
    (\Delta(r_0)+\xi(\lambda))(1+d_\star)
    &
    d_{jk},d_{ij}\le r_0\,,\\
    d_\star \Delta(d_{jk})
    &
    d_{jk}>0,d_{ij}\le r_0\,.
    \end{cases}
\end{align}
Then we proceed similarly regarding the second desired inequality, and bound
\begin{align}
    \|\delta_i^\alpha(\mathcal{L}_i-\mathcal{L}_i^0)(O)\|
    &\le
    \|\delta_i^\alpha(\mathcal{L}_i-\mathcal{L}_i^{r_0})(O)\|
    +
    \|\delta_i^\alpha(\mathcal{L}_i^{r_0}-\mathcal{L}_i^0)(O)\|
    \\
    &\le 
    d_\star 
    \sum_{\alpha'}
    \sum_k \Delta(\max(r_0,d_{ik}))
    \|
    \delta_{k}^{\alpha'}(O)
    \|
    +
    d_\star 
    \xi(\lambda)
    \sum_{\alpha'}
    \sum_{k|d_{ik}\le r_0} 
    \|
    \delta_{k}^{\alpha'}(O)
    \|\,,
\end{align}
from which we can read off the values of $\gamma$'s as
\begin{align}
    \gamma_i^{k,\alpha'}
    =
    \begin{cases}
       d_\star (\Delta(r_0)+\xi(\lambda)) & d_{ik}\le r_0\,,\\
       d_\star \Delta(d_{ik}) & d_{ik}> r_0\,.\\
    \end{cases}
\end{align}
Note that we have chosen both $\kappa^{k\alpha}_{ij}$ and $\gamma_i^{k\alpha}$ to be independent of $\alpha$.
Finally, we can sum these up as follows:
\begin{align}
        \sum_{i}
        \Big(
        \sum_{j\neq i}
        \kappa_{ij}^{\ell,\alpha}
        +
        \gamma_i^{\ell,\alpha}
        \Big)
        &\le 
        \sum_{j(\neq k) |d_{jk}>r_0}
        \Delta(d_{jk})
        +
        d_\star 
        \sum_{i\neq j|d_{ij}>r_0}
        \Delta(\max(d_{ij},d_{jk}))
        \\
        &\quad 
        +
        (\Delta(r_0)+\xi(\lambda))(1+d_\star)
        \sum_{i\neq j|d_{ij},d_{jk}\le r_0}1
        +
        d_\star 
        \sum_{i\neq j|d_{ij}\le r_0,d_{jk}> r_0}
        \Delta(d_{jk})
        \\
        &\quad +
        d_\star 
        (\Delta(r_0)+\xi(\lambda))
        \sum_{i|d_{ik}\le r_0}
        1
        +
        d_\star 
        \sum_{i|d_{ik}> r_0}
        \Delta(d_{ik})
        \,.
    \end{align}   
    Next we are going to use the bound $|B_j(r)|,|\{i|d_{ij}=r\}|
    \le (2r+1)^D$.
    In \cite{rouze2024optimalquantumalgorithmgibbs} they used the bound $|\{i|d_{ij}=r\}|
    \le (2r+1)^{D-1}$ but this does not hold at $D=1$, so we instead use this looser bound which holds for any $D$.
    We use the notation $I$ for the indicator function and label the summands above as $(1)$ to $(6)$, and we bound each of them separately in a system-size-independent manner:
\begin{align}
&(1)
    \sum_{j(\neq k) |d_{jk}>r_0}
        \Delta(d_{jk})
    =
    \sum_{\ell > r_0}
    \Delta(\ell) 
    |\{ 
    j|d_{jk}=\ell
    \}|
    \le 
    \sum_{\ell > r_0}
    \Delta(\ell) 
    (2\ell+1)^D\\
&(2)
        \sum_{i\neq j|d_{ij}>r_0}
        \Delta(\max(d_{ij},d_{jk}))
\le 
\sum_{d_{jk}>d_{ij}>r_0}
\Delta(d_{jk}))
+ 
\sum_{d_{ij}>r_0, d_{jk}\le d_{ij}}
\Delta(d_{ij}))
\\
&\quad\le 
\sum_{j|d_{jk}>r_0}
\Delta(d_{jk})
|\{i|d_{ij}\le d_{jk},
d_{ij}>r_0\}|
+
\sum_j \sum_{\ell \ge r_0,d_{jk}}
\Delta(\ell) (2\ell+1)^D
\\
&\quad\le 
\sum_{j|d_{jk}>r_0}
\Delta(d_{jk})
|\{i|d_{ij}\le d_{jk}\}|
+
\sum_{\ell'} 
\sum_{\ell \ge r_0,\ell'}
\Delta(\ell) (2\ell+1)^D
(2\ell'+1)^D
\\
&\quad\le 
\sum_{\ell>r_0}
\Delta(\ell)
|\{i|d_{ij}\le \ell\}| \cdot
|\{j|d_{jk} = \ell\}|
+
r_0
\sum_{\ell \ge r_0}
\Delta(\ell) (2\ell+1)^{2D}
+
\sum_{\ell'\ge r_0}
\sum_{\ell \ge \ell'}
\Delta(\ell) (2\ell+1)^{2D}
\\
&\quad\le 
(1
+
r_0)
\sum_{\ell \ge r_0}
\Delta(\ell) (2\ell+1)^{2D}
+
\sum_{\ell'\ge r_0}
\sum_{\ell \ge \ell'}
\Delta(\ell) (2\ell+1)^{2D}
\end{align}
\begin{align}
&(3)
    \sum_{i\neq j|d_{ij},d_{jk}\le r_0}1
    =
    \sum_j
    \sum_{i(\neq j)}
    I(d_{ij}\le r_0)
    I(d_{jk}\le r_0)
    \le 
    (2r_0+1)^{2D}
\\
&(4)\sum_{i\neq j|d_{ij}\le r_0,d_{jk}> r_0}
        \Delta(d_{jk})
    =
    (2r_0+1)^D
    \sum_{\ell>r_0}
    \Delta(\ell)
    (2\ell+1)^D\\
&(5)
   \sum_{i|d_{ik}\le r_0}
        1
        \le (2r_0+1)^D
    \\ 
        &(6)\sum_{i|d_{ik}> r_0}
        \Delta(d_{ik})
        =
        \sum_{\ell>r_0}
        \Delta(\ell)
        (2\ell+1)^D
\end{align}
Putting things together we get
\begin{align}\hspace{-0.2cm}
\sum_{i}
        \Big(
        \sum_{j\neq i}
        \kappa_{ij}^{\ell,\alpha}
        +
        \gamma_i^{\ell,\alpha}
        \Big)
        &\le 
        \sum_{\ell > r_0}
    \Delta(\ell) 
    (2\ell+1)^D
    +d_\star 
    \Big((1
+
r_0)
\sum_{\ell \ge r_0}
\Delta(\ell) (2\ell+1)^{2D}
+
\sum_{\ell'\ge r_0}
\sum_{\ell \ge \ell'}
\Delta(\ell) (2\ell+1)^{2D} \Big)\hspace{1cm}
\\
&\quad +
(\Delta(r_0)+\xi(\lambda))(1+d_\star)(2r_0+1)^{2D}+
d_\star 
(2r_0+1)^D
    \sum_{\ell>r_0}
    \Delta(\ell)
    (2\ell+1)^D
    \\
&\quad +
d_\star 
        (\Delta(r_0)+\xi(\lambda)) (2r_0+1)^D
+
d_\star 
\sum_{\ell>r_0}
        \Delta(\ell)
        (2\ell+1)^D
    \\
&\le 
(\Delta(r_0)+\xi(\lambda))(1+2d_\star)(2r_0+1)^{2D}
+
\left(d_\star(2+r_0) + d_\star (2r_0+1)^D+1 \right)
\sum_{\ell>r_0}
        \Delta(\ell)
        (2\ell+1)^{2D}\hspace{-2cm}
\\
&\quad + d_\star 
\sum_{\ell'\ge r_0}
\sum_{\ell \ge \ell'}
\Delta(\ell) (2\ell+1)^{2D}
\equiv 
\xi(\lambda)(1+2d_\star)(2r_0+1)^{2D}
+f(r_0)\equiv \eta \,,
\end{align}
giving us the system-size-independent $\eta$, where we have defined
\begin{align}
f(r_0)=
&\ \Delta(r_0)(1+2d_\star)(2r_0+1)^{2D}
+
\left(d_\star(2+r_0) + d_\star (2r_0+1)^D+1 \right)
\sum_{\ell>r_0}
        \Delta(\ell)
        (2\ell+1)^{2D}
\\
&\quad + d_\star 
\sum_{\ell'\ge r_0}
\sum_{\ell \ge \ell'}
\Delta(\ell) (2\ell+1)^{2D}
\,.
\end{align}
The existence of this $\eta$ will be crucial for extending the rapid mixing from the non-interacting systems to the weakly interacting systems, conditioned on $\eta$ being sufficiently small, which will in turn give an upper bound on the maximal coupling strength $\lambda_\textup{max}$.
Note that this $\eta$ coincides with the expression in \cite[Eqs.~A35-A36]{rouze2024optimalquantumalgorithmgibbs}
if we set $d_\star =2$ and replace the exponent $2D$ with $2D-1$ and $2D-2$ as coming from the different bounds on the number of points on the surface of a ball on the lattice, as remarked above.
\end{proof}

Note that for the case at hand, Lemma \ref{lemma:bounds on delta L} implies that
\begin{align}
    \label{eq:Delta ell}
    \Delta(\ell) = C
    \sum_{r\ge \ell}\e^{-\chi r}
    =
    C
    \e^{-\chi \ell}
    \sum_{r\ge 0}\e^{-\chi r}
    =
    \frac{C}{1-\e^{-\chi}}\e^{-\chi \ell}
    \equiv
    \tilde{C}
    \e^{-\chi \ell}\,.
\end{align}


\subsection{Rapid mixing results}\label{sec: rapid mixing of qudits}

After these preparations, we can prove rapid mixing of the algorithmic Lindbladian associated with weakly interacting spin systems.

\begin{theorem}\label{thm:rapid mixing separable general}
Consider a system of $n$ qudits of dimension $d$.
Let $\mathcal{L}^0$ be the algorithmic Lindbladian associated to the separable Hamiltonian $H^0=\sum_ih_i$ and $\mathcal{L}$ that associated to $H=H^0+\lambda V$, where $V$ has at least exponentially decaying interactions and  the jump operators of $\mathcal{L}^0$ and $\mathcal{L}$ are strictly local.
Denote by $\Delta_0$ the spectral gap of 
$\mathcal{L}^0$ and by $\Delta E_\star=\max_i \Delta E_i$, where $\Delta E_i$ is the difference between the maximum and minimum eigenvalue of the Hamiltonian $h_i$. Then, as long as $d^2 \eta < \Delta_0$, where
$\eta$ is the constant introduced in Lemma \ref{lemma:eta},
$\mathcal{L}$ mixes rapidly:
\begin{align}
        t_{\textup{mix}}(\epsilon)
        \le 
        \frac{1}{\Delta_0 - d^2\eta}
        \log\left(
        4
        d^4 \e^{2\beta \Delta E_\star}
        \frac{n}{\epsilon} \right)
        \,.
    \end{align}
\end{theorem}

\begin{proof}
    First recall the Definition of the oscillator norm \ref{def:original oscillator norm} and adapt it to the qudit case with dimension $d$ by defining $\delta_i(A) = A - C_i(A)$, where
    $C_I(A)=d^{-|I|}\Tr_{I}(A)$ is the normalised trace on subsystem $I$.
    Then we have as previously
    \begin{equation}
        \|\rho(t)-\sigma\|_{\Tr} 
        = \sup_{\|O\|\leq 1} \|O(t) - \Tr(O(t))/d^n\| \cdot \|\rho - \sigma\|_{\Tr}
        \le 
        2
        \sup_{\|O\|\leq 1}
        \sum_{i=1}^n 
        \|\delta_i(O(t))\| =  2
        \sup_{\|O\|\leq 1} \oscnorm{O(t)}\,.
    \end{equation}

    Now note that $\mathcal{L}^0 = \sum_{i=1}^n \mathcal{L}^0_i$ where $\mathcal{L}_i^0$ is supported only on qudit $i$.
    Recall from Lemma \ref{lemma:properties Falpha} that $\mathcal{L}^0_i$ has a basis of eigenvectors $F^\alpha_i$ with eigenvalues $\lambda_i^\alpha$ which are orthonormal w.r.t.~the KMS inner product with the state $\sigma^0_i=
    \e^{-\beta h_i}/Z_i^0$, so that we can expand the observable $O$ as
    \begin{align}
        O
        =
        \sum_{\alpha=1}^{d^2}
        F_i^\alpha
        \langle F_i^\alpha, O\rangle_{\sigma^0_i}
    \end{align}
    and define the corresponding partial quasi-derivations
    \begin{align}
        \delta_i^\alpha(O)
        =
        \delta_i( 
        F_i^\alpha)
        \langle 
        F_i^\alpha,O \rangle_{\sigma^0_i}\,.
    \end{align}
    Note that 
    $\delta_i^1(O) = 0$, and so $
        \sum\limits_{\alpha=2}^{d^2}
        \delta_i^\alpha(O)
        =
        \delta_i(O)$, as well as 
        $\delta_i^\alpha
        (F^{\alpha'}_i)
        =
        \delta_{\alpha,\alpha'}
        \delta_i(F^\alpha)$.
    The first equation follows since $F_i^1=\id$ and $\delta_i(\id)=0$.
    We can hence bound the original oscillator norm as
    \begin{align}
        \sum_{i=1}^n\| \delta_i(O)\|
        \le 
        \sum_{\alpha=2}^{d^2}
        \sum_{i=1}^n
        \| \delta^\alpha_i(O)\|\,.
    \end{align}
    The usefulness of the quasi-derivation $\delta_i^\alpha$ is that it satisfies the following relation:
    \begin{align}
        \label{eq:closure delta alpha}
        \delta_i^\alpha 
        \mathcal{L}_i^0 O
        =
        \delta_i^\alpha 
        \sum_{\alpha'}
        \lambda_i^{\alpha'}
        F_i^{\alpha'}
        \langle 
        F_i^{\alpha'},
        O
        \rangle_{\sigma^0_i}
        =
        \sum_{\alpha'}
        \lambda_i^{\alpha'}
        \langle 
        F_i^{\alpha'},
        O
        \rangle_{\sigma^0_i}
        \delta_i^{\alpha}(F_i^{\alpha'})
        =
        \lambda_i^{\alpha}
        \langle 
        F_i^{\alpha},
        O
        \rangle_{\sigma^0_i}
        \delta_i(F_i^{\alpha})
        =
        \lambda_i^{\alpha}
        \delta_i^\alpha(O)
        \,,
    \end{align} which is a key ingredient for bounding the decay of $\oscnorm{O(t)}$.
    The remainder of the proof largely follows \cite{Majewski1995} and \cite{rouze2024optimalquantumalgorithmgibbs}, and we report it here for the reader's convenience.
    Let $\mathcal{P}^t=\e^{t\mathcal{L}}$, 
    $\mathcal{P}^t_{\neq i}=\e^{t\sum_{j\neq i}\mathcal{L}_j}$.
    Then
    \begin{align}\hspace{-0.3cm}
        \frac{\partial}{\partial s}\Big( 
        \e^{-\lambda_i^\alpha s}
        \mathcal{P}^{t-s}_{\neq i}
        &\delta_i^\alpha \mathcal{P}^s O
        \Big)
        =
        -\lambda_i^\alpha
        \e^{-\lambda_i^\alpha s}
        \mathcal{P}^{t-s}_{\neq i}
        \delta_i^\alpha O(s)
        +
        \e^{-\lambda_i^\alpha s}
        \mathcal{P}^{t-s}_{\neq i}
        \Big( 
        -\sum_{j\neq i}\mathcal{L}_j
        \Big)
        \delta_i^\alpha O(s)
        +
        \e^{-\lambda_i^\alpha s}
        \mathcal{P}^{t-s}_{\neq i}
        \delta_i^\alpha 
        \Big( 
        \sum_{j}\mathcal{L}_j
        \Big)
        O(s)\hspace{-0.3cm}
        \\
        &=
        -
        \e^{-\lambda_i^\alpha s}
        \mathcal{P}^{t-s}_{\neq i}
        \delta_i^\alpha 
        \mathcal{L}_i^0
        O(s)
        +
        \e^{-\lambda_i^\alpha s}
        \mathcal{P}^{t-s}_{\neq i}
        \Big( 
        -\sum_{j\neq i}\mathcal{L}_j
        \Big)
        \delta_i^\alpha O(s)
        +
        \e^{-\lambda_i^\alpha s}
        \mathcal{P}^{t-s}_{\neq i}
        \delta_i^\alpha
        \Big( 
        \sum_{j\neq i}\mathcal{L}_j
        +
        \mathcal{L}_i
        \Big)
        O(s)
        \\
        &=
        \e^{-\lambda_i^\alpha s}
        \mathcal{P}^{t-s}_{\neq i}
        \delta_i^\alpha 
        (
        \mathcal{L}_i
        -
        \mathcal{L}_i^0)
        O(s)
        +
        \e^{-\lambda_i^\alpha s}
        \mathcal{P}^{t-s}_{\neq i}
        \Big[
        \delta_i^\alpha,
        \sum_{j\neq i}\mathcal{L}_j
        \Big]
         O(s)\,,
    \end{align}
    where in the second line we used relation \eqref{eq:closure delta alpha}.
    Now we integrate this equation from $0$ to $t$:
    \begin{align}
        \e^{-\lambda_i^\alpha t}
        \delta_i^\alpha O(t)
        -
        \mathcal{P}^{t}_{\neq i}
        \delta_i^\alpha O
        =
        \int_0^t \dd s\,
        \e^{-\lambda_i^\alpha s}
        \mathcal{P}^{t-s}_{\neq i}
        \Big( 
        \delta_i^\alpha 
        (
        \mathcal{L}_i
        -
        \mathcal{L}_i^0)
        O(s)
        +
        \Big[
        \delta_i^\alpha,
        \sum_{j\neq i}\mathcal{L}_j
        \Big]
         O(s)
        \Big)
    \end{align}
    and take the operator norm, using that $\mathcal{P},\mathcal{P}_{\neq i}$ are contractions 
    (see e.g.~\cite{P_rez_Garc_a_2006}):
    \begin{align}
        \|\delta_i^\alpha O(t)\|
        &\le
        \e^{\lambda_i^\alpha t}
        \|\mathcal{P}^{t}_{\neq i}
        \delta_i^\alpha O\|
        +
        \int_0^t \dd s\,
        \e^{\lambda_i^\alpha (t-s)}
        \|
        \mathcal{P}^{t-s}_{\neq i}
        \Big(
        \delta_i^\alpha 
        (
        \mathcal{L}_i
        -
        \mathcal{L}_i^0)
        O(s)
        +
        \sum_{j\neq i}
        [
        \delta_i^\alpha,
        \mathcal{L}_j
        ]
        O(s)
        \Big)
        \|\\
        &\le 
        \e^{\lambda_i^\alpha t}
        \|
        \delta_i^\alpha O\|
        +
        \int_0^t \dd s\,
        \e^{\lambda_i^\alpha (t-s)}
        \|
        \delta_i^\alpha 
        (
        \mathcal{L}_i
        -
        \mathcal{L}_i^0)
        O(s)
        \|
        +
        \sum_{j\neq i}\|
        [
        \delta_i^\alpha,
        \mathcal{L}_j
        ]
        O(s)\|
        \Big)\,.
    \end{align}
    At this point we use 
    Lemma \ref{lemma:bounds on delta L} and \ref{lemma:eta} so that there exist $\kappa_{ij}^{\ell,\alpha}, \gamma_i^{\ell,\alpha},\eta\ge 0$ such that, denoting $\lambda^\alpha_\star=\max_i\lambda_i^\alpha$, and by 
    $\Delta_0 = -\max_{\alpha>1,i}\lambda^\alpha_i$ the spectral gap of the Lindbladian $\mathcal{L}^0$, we have
    \begin{align}
        \sum_{\alpha>1}\sum_i 
        \|\delta_i^\alpha O(t)\|
        &\le
        \sum_{\alpha>1}\e^{\lambda_\star^\alpha t}
        \sum_i 
        \|
        \delta_i^\alpha O\|
        +
        \sum_{\alpha>1}
        \int_0^t \dd s\,
        \e^{\lambda_\star^\alpha (t-s)}
        \sum_{\alpha'>1}
        \sum_\ell        
        \sum_i 
        \Big(
        \gamma_i^{\ell,\alpha'}
        +
        \sum_{j\neq i}
        \kappa_{ij}^{\ell,\alpha'}
        \Big)\| \delta_\ell^{\alpha'}
        O(s)\|
        \\
        &\le 
        \sum_{\alpha>1}\e^{\lambda_\star^\alpha t}
        \sum_i 
        \|
        \delta_i^\alpha O\|
        +
        \eta
        \sum_{\alpha>1}
        \int_0^t \dd s\,
        \e^{\lambda_\star^\alpha (t-s)}
        \sum_{\alpha'>1}
        \sum_\ell        
        \| \delta_\ell^{\alpha'}O(s)\|\\
        &\le 
        \e^{-\Delta_0  t}
        \sum_{\alpha>1}
        \sum_i 
        \|
        \delta_i^\alpha O\|
        +
        \eta d^2
        \int_0^t \dd s\,
        \e^{-\Delta_0 (t-s)}
        \sum_{\alpha'>1}
        \sum_\ell        
        \| \delta_\ell^{\alpha'}O(s)\|
        \,.
    \end{align}    
    Now this equation is of the form
    \begin{align}
    \e^{-\lambda t} f(t) \le    
    f(0)
    +
    \eta
    \int_0^t \dd s\,
    \e^{-\lambda s}
    f(s)    
    \end{align}
    with 
    $f(t) = \sum_i 
    \|\delta_i^\alpha O(t)|$.
    Taking the derivative w.r.t.~$t$ yields
    \begin{align}
        -\lambda \e^{-\lambda t} f(t)
        +
        \e^{-\lambda t} \dot{f}(t)
        \le \eta e^{-\lambda t} f(t)\,,
    \end{align}
    which from Gr\"onwall's inequality gives
    $f(t)\le \e^{(\lambda + \eta)t} f(0)$, and in our case
    \begin{align}
        \sum_{\alpha>1}
        \sum_i 
        \|\delta_i^\alpha O(t)\|
        \le 
        \e^{(-\Delta_0 + \eta d^2) t}
        \sum_{\alpha>1}
        \sum_i 
        \|\delta_i^\alpha O\|\,.
    \end{align}
    Since $C_i$ is a contraction and using the calculations of Lemma \ref{lemma:properties Falpha}, we have
    \begin{align}
        \|\delta_i^\alpha O\|
        &=
        \|
        \langle 
        F_i^\alpha,O \rangle_{\sigma^0_i}
        \delta_i( 
        F_i^\alpha)
        \|
        \le 
        2\| \langle 
        F_i^\alpha,O \rangle_{\sigma^0_i}
        F_i^\alpha \|
        \le 
        2
        \|  
        (F_i^\alpha)^\dagger
        (\sigma_i^0)^{1/2}O
        (\sigma_i^0)^{1/2}
        \|
        \|
        F_i^\alpha \|
        \le 
        2
        d^2 \e^{2\beta \Delta E_i} \| O \|\,,
    \end{align}
    so that finally
    \begin{align}
        \sup_{\|O\|\le 1}\sum_{\alpha>1}
        \sum_i 
        \|\delta_i^\alpha O(t)\|
        \le 
        d^2 d_\star
        n
        \e^{(-\Delta_0 + \eta d^2) t}
        \,,\quad 
        d_\star =2
        d^2 \e^{2\beta \Delta E_\star}\,,
    \end{align}
    where $\Delta E_\star = \max_i 
    \Delta E_i$ and thus
    \begin{align}
        \|\rho(t)-\sigma\|_{\Tr}
        \le 
        4
        d^4 \e^{2\beta \Delta E_\star}
        n
        \e^{(-\Delta_0 + \eta d^2) t}\,.
    \end{align}
    For this to be decaying exponentially in time, we need 
    $d^2 \eta <\Delta_0$.
    Setting the r.h.s.~of the equation above smaller than $\epsilon$ and solving for $t$, we get the logarithmic bound on the mixing time
    \begin{align}
        t_{\textup{mix}}(\epsilon)
        \le 
        \frac{1}{\Delta_0 - d^2\eta}
        \log\left(
        4
        d^4 \e^{2\beta \Delta E_\star}
        \frac{n}{\epsilon} \right)\,.
    \end{align}
\end{proof}

The condition $d^2\eta < \Delta_0$ corresponds to the following bound on the coupling strength $\lambda$:

\begin{corollary}\label{cor:maximal lambda}
The Lindbladian from Theorem \ref{thm:rapid mixing separable general} mixes rapidly if 
\begin{align}
    \label{eq:lambda bound}
   |\lambda| 
   <
   \lambda_\textup{max} \equiv 
   \frac{\Delta_0 - d^2 f(r_0)}{d^2 (1+ 4d^2\e^{2\beta \Delta E_\star})
   (2r_0+1)^{2D} C'}
   \,.
\end{align}
Here, $f(r_0)$ is defined in \eqref{eqn:f_decay function} together with \eqref{eq:Delta ell}, $C'$ in Lemma \ref{lemma:bounds on delta L}, and $r_0$ is a positive parameter that controls the bounds in Lemma \ref{lemma:eta}, which can be always chosen such that $\lambda_\textup{max}>0$.
\end{corollary}

Note that in \eqref{eq:lambda bound} $C'$ depends on $\beta$ and the decay rate of the interactions in the perturbation,  $f(r_0)$
depends on $\beta$ and the parameters occurring in the Lieb-Robinson bound and $\Delta_0$ depends on the non-interacting Hamiltonian $H_0$, the choice of jump operators and the filter function -- which in turns depends on $\beta$. 
For a given Hamiltonian, we can use this bound to optimise $r_0$ and the choice of jump operators and filter functions, so that the range of perturbation strengths which correspond to rapid mixing is maximised.
Note that there always exists a choice of $r_0$ such that the numerator is positive for the following reason. For the choice of $\Delta(l)$ in \eqref{eq:Delta ell}, there exist positive constants $C,C'$ such that:
\begin{align}    
    \sum_{\ell>r_0}
    \Delta(\ell)(2\ell+1)^{2D}
    &=
    \tilde{C}
    \sum_{y>0}
    \e^{-\chi(y+r_0)}(2r_0+2y+1)^{2D}
    =
    C
    \e^{-\chi r_0}
    r_0^{2D}
    (1+\mathcal{O}(r_0^{-1}))
\\
    \sum_{\ell'\ge r_0}
\sum_{\ell \ge \ell'}
\Delta(\ell) (2\ell+1)^{2D}
&=
\tilde{C}
\sum_{y,z\ge 0}
\e^{-\chi(y+z+r_0)}(2r_0+2y+2z+1)^{2D}
=
C'
    \e^{-\chi r_0}
    r_0^{2D}
    (1+\mathcal{O}(r_0^{-1}))
\end{align}
so that for large $r_0$ the bound on $\lambda$ is of the form $\Delta_0 r_0^{-2D} - \e^{-\chi r_0} p(r_0)$ up to an overall constant independent of $r_0$, where $p(r_0)$ is a polynomial in $r_0$. Thus, for any values of the parameters appearing in the bound, there is a finite $r_0$ such that $\Delta_0 r_0^{-2D} - \e^{-\chi r_0} p(r_0)$ is positive since the second negative term goes to zero  faster than the first term.\\

While the presentation here focused on exponentially decaying interactions, we remark that it is possible to extend these results for slowly decaying power-law interactions as done in the high-temperature case in \cite{rouze2024optimalquantumalgorithmgibbs}.
Using formula \eqref{eq:runtime_quantum} the results presented in this section directly translate to a $\tilde{\mathcal{O}}(n^2)$ Hamiltonian simulation time to prepare the Gibbs state of weakly-interacting spin systems.


\section{Fermionic Systems}\label{sec:fermions}

\subsection{Oscillator norms for fermions}

This section will adapt the oscillator norm tools to the case of fermionic Lindbladians. We shall work with canonical fermionic creation and annihilation operators $c^\dagger, c$ obeying $\{c_i,c_j^\dagger\} = \delta_{ij}$. We will follow the definitions and theory for fermions developed in \cite{zhan2025rapidquantumgroundstate}.

\begin{definition}\label{def:fermionic trace}
    The \textit{fermionic quasi-derivation} $\delta^f_a$ at site $a$ is defined by \begin{equation}
        \delta^f_a(O) = O - \frac{c_a^\dagger c_a + c_a c_a^\dagger}{2} \Tr^f_a(O) = O - \frac{I}{2} \Tr^f_a(O)\,,
    \end{equation} where $\Tr^f_a$ denotes \textit{fermionic partial trace} over site $a$. This can be defined by a sum-product formula \begin{equation}
        \Tr_a^f(O) = \frac{1}{2}\left(O + (c_a+c_a^\dagger) O (c_a+c_a^\dagger) + (c_a-c_a^\dagger) O (c_a^\dagger - c_a) + (1-2c_a^\dagger c_a) O (1-2c_a^\dagger c_a)\right)\,.\label{eqn:fermionic trace}
    \end{equation}
\end{definition}

\begin{definition}\label{def:fermionic osci norm}
    The \textit{projected fermionic oscillator norm} is defined by \begin{equation}
        \oscnorm{O} = \sum_{i=1}^n \| \delta_i^f(P_i(O))\| + \| \delta_i^f(Q_i(O))\|\,,
    \end{equation} where $P_i$ is the projector onto the diagonal space at site $i$ and $Q_i$ onto the off-diagonal space, given by \begin{align}
        P_i(O) &= c_i c_i^\dagger O c_i c_i^\dagger + c_i^\dagger c_i O c_i^\dagger c_i\,,\\
        Q_i(O) &= c_i c_i^\dagger O c_i^\dagger c_i + c_i^\dagger c_i O c_i c_i^\dagger\,.
    \end{align}
\end{definition}

\begin{remark}
    The previous definitions are given with respect to some specific set of canonical fermions ${c_i}$, but later we shall consider redefining them in terms of a different canonical set, specific to the given Lindbladian.
\end{remark}

\begin{remark}
    The projected oscillator norm is an upper bound to the non-projected one defined in \ref{def:original oscillator norm} as \begin{equation}
         \oscnorm{O} = \sum_{i=1}^n \| \delta_i^f(P_i(O))\| + \| \delta_i^f(Q_i(O))\| \geq  \sum_{i=1}^n \| \delta_i^f(P_i(O)+Q_i(O))\| =  \sum_{i=1}^n \| \delta_i^f(O)\|\,.
    \end{equation}
\end{remark}

\begin{lemma}
    We can bound the trace distance between a time-evolved state $\rho(t)$ and the Gibbs state $\sigma$ (i.e. the fixed state of the evolution) using this oscillator norm like \begin{equation}
        \|\rho(t)-\sigma\|_{\Tr} \leq \sup_{\|O\|\leq 1} \oscnorm{O(t)} \cdot \|\rho(0)-\sigma\|_{\Tr}  \leq 2 \cdot \sup_{\|O\|\leq 1} \oscnorm{O(t)}\,,
    \end{equation} where $\rho(t) = e^{t\mathcal{L}^\dagger}[\rho(0)]$ and $O(t) = e^{t\mathcal{L}}[O]$.
\end{lemma}
This holds equally as before for the case of fermions when we restrict ourselves to the subspace of even observables. Hence, since we can now bound $\oscnorm{O} \leq 4n$ for any $O$ obeying $\|O\| \leq 1$, it again suffices to show that $\oscnorm{O(t)} \leq e^{-\alpha t} \oscnorm{O(0)}$ for some $\alpha > 0$ to prove rapid mixing of the fermionic Lindbladian $\mathcal{L}$.


\subsection{Free fermions}\label{sec: rapid mixing of free fermions}

Consider a free fermionic Hamiltonian
\begin{align}
H_0 = \sum_{i,j} M_{ij} c_i ^\dagger c_j,
\end{align}
where $M$ is a hermitian $n\times n$ matrix, which we will also assume to be real for simplicity (so that it is symmetric). Note that we can diagonalise $M$ like $M = U^T D U$, with $U$ being orthogonal (real and unitary), where $D = \operatorname{diag}(\epsilon_1, \dots \epsilon_{n})$ are the eigenvalues of $M$. Consider the Bogoliubov transformation $\mathbf{b} = U \cdot \mathbf{c}$. These operators also form a set of canonical fermions, with respect to which the free fermionic Hamiltonian becomes $H_0 = \sum_i \epsilon_ib_i^\dagger b_i$.\\

Similarly to \cite{smid2025polynomial,Tong_2025Fast}, choose the jump operators of the Gibbs sampler to be $A_i^{(1)} = c_i$ and $A_i^{(2)} = c_i^\dagger$ and the filter functions $\hat{f}^a = \hat{f}$ to be equal and real (note that we can choose non-Hermitian jumps if we also include their conjugates and choose the filter functions of the pair to be the same). The non-interacting Lindbladian in the Heisenberg picture is then \begin{align}\hspace{-0.2cm}
    \mathcal{L}_0[O] &= \sum_{a=1}^n L_{0,a}^{(1)\dagger} O L^{(1)}_{0,a} - \frac{1}{2}\{L_{0,a}^{(1)\dagger} L^{(1)}_{0,a},O\} + \sum_{a=1}^n L_{0,a}^{(2)\dagger} O L^{(2)}_{0,a} - \frac{1}{2}\{L_{0,a}^{(2)\dagger} L^{(2)}_{0,a},O\}\\
    &= \sum_{a=1}^n \sum_{s,q} \hat{f}(-M)_{as}\hat{f}(-M)_{aq} \cdot \left( c_s^\dagger Oc_q -\frac{1}{2} \{c_s^\dagger c_q,O\} \right) +  \sum_{a=1}^n \sum_{s,q} \hat{f}(M)_{as}\hat{f}(M)_{aq} \cdot \left( c_s Oc_q^\dagger -\frac{1}{2} \{c_s c_q^\dagger,O\} \right)\hspace{-0.2cm}\\
     &= \sum_{i=1}^n \hat{f}(-\epsilon_i)^2 \cdot \left(b_i^\dagger O b_i - \frac{1}{2} \{b_i^\dagger b_i, O\}\right) + \hat{f}(\epsilon_i)^2 \cdot \left(b_i O b_i^\dagger - \frac{1}{2} \{b_i b_i^\dagger, O\}\right) \equiv \sum_{i=1}^n \mathcal{L}_i^{(1)}[O] + \mathcal{L}_i^{(2)}[O]\,,
\end{align} where on the last line we have diagonalised the Lindbladian using the Bogoliubov transformation. Note that this is the same form we would get if we were to consider the diagonalised Hamiltonian $H_0 = \sum_i \epsilon_ib_i^\dagger b_i$ and take the jump operators to be the set $\{b_i,b_i^\dagger\}_{i=1}^n$. This happens because any Lindbladian is invariant under a unitary transformation of the Lindblad operators, $L_a \to L_a ' = \sum_a U_{ab}L_b$, and so any unitary transformation of the jump operators actually leads to the same Lindbladian.\\

Now, we define the fermionic oscillator norm (Definitions \ref{def:fermionic trace}, \ref{def:fermionic osci norm}) in terms of the transformed Bogoliubov fermions $b_i$. Consider a general even observable $O$, and expand it along the $i$th site, so that \begin{equation}
    O = \id \otimes~O^i_{0,0} + b_i^\dagger \otimes O^i_{1,0} + b_i \otimes O^i_{0,1} + b_i^\dagger b_i \otimes O^i_{1,1}\,,
\end{equation} where $\otimes$ denotes the fermionic graded tensor product with a fixed ordering. Then we find
\begin{equation}
    \delta_i^f(O) = \frac{1}{2}(b^\dagger_i b_i - b_i b_i^\dagger) \otimes O^i_{1,1} + b_i^\dagger \otimes O^i_{1,0} + b_i \otimes O^i_{0,1}\,,
\end{equation}
as well as
\begin{align}
    \mathcal{L}_i^{(1)}[O] &= \hat{f}(-\epsilon_i)^2 \cdot \left(-b^\dagger_i b_i \otimes O^i_{1,1} - \frac{1}{2} b_i^\dagger \otimes O^i_{1,0} - \frac{1}{2} b_i \otimes O^i_{0,1}\right)\\
     \mathcal{L}_i^{(2)}[O] &= \hat{f}(\epsilon_i)^2 \cdot \left(b_i b_i ^\dagger \otimes O^i_{1,1} - \frac{1}{2} b_i^\dagger \otimes O^i_{1,0} - \frac{1}{2} b_i \otimes O^i_{0,1}\right)\,.
\end{align}
Now we can evaluate
\begin{align}
    \delta_i^f(\mathcal{L}_i^{(1)}[O]) &= \hat{f}(-\epsilon_i)^2 \cdot \left(\frac{1}{2}(b_i b_i^\dagger - b_i^\dagger b_i) \otimes O^i_{1,1}  - \frac{1}{2} b_i^\dagger \otimes O^i_{1,0} - \frac{1}{2} b_i \otimes O^i_{0,1}\right)\\
    \delta_i^f(\mathcal{L}_i^{(2)}[O]) &= \hat{f}(\epsilon_i)^2 \cdot \left(\frac{1}{2}(b_i b_i^\dagger - b_i^\dagger b_i) \otimes O^i_{1,1}  - \frac{1}{2} b_i^\dagger \otimes O^i_{1,0} - \frac{1}{2} b_i \otimes O^i_{0,1}\right)\,,
\end{align}
so that together
\begin{equation}
     \delta_i^f(\mathcal{L}_i[O]) = 2q(\epsilon_i)^2\cosh(\beta \epsilon_i /2) \cdot \left(\frac{1}{2}(b_i b_i^\dagger - b_i^\dagger b_i) \otimes O^i_{1,1}  - \frac{1}{2} b_i^\dagger \otimes O^i_{1,0} - \frac{1}{2} b_i \otimes O^i_{0,1}\right)\,.
\end{equation}
This is important because now we can evaluate the relations between the corresponding actions on projected observables as follows:
\begin{align}
    \delta_i^f(P_i(\mathcal{L}_i[O])) &= - 2q(\epsilon_i)^2\cosh(\beta \epsilon_i /2) \cdot \delta_i^f(P_i(O))\\
    \delta_i^f(Q_i(\mathcal{L}_i[O])) &= - q(\epsilon_i)^2\cosh(\beta \epsilon_i /2) \cdot \delta_i^f(Q_i(O))\,.
\end{align} 
Hence we can follow as previously and find that 
\begin{align}
    \frac{\dd}{\dd t} \delta_i^f(P_i(O(t))) &= \delta_i^f(P_i(\mathcal{L}_i^0[O(t)]))+ \sum_{j \neq i} \delta_i^f(P_i(\mathcal{L}^0_j [O(t)]))\\ 
    &= - 2q(\epsilon_i)^2\cosh(\beta \epsilon_i /2) \cdot \delta_i^f(P_i(O(t)))  + \sum_{j \neq i} \mathcal{L}_j [\delta_i^f(P_i(O(t)))]\,,
\end{align}
and similarly
\begin{align}
    \frac{\dd}{\dd t} \delta_i^f(Q_i(O(t))) &= \delta_i^f(Q_i(\mathcal{L}_i^0[O(t)]))+ \sum_{j \neq i} \delta_i^f(Q_i(\mathcal{L}^0_j [O(t)]))\\ 
    &= - q(\epsilon_i)^2\cosh(\beta \epsilon_i /2) \cdot \delta_i^f(Q_i(O(t)))  + \sum_{j \neq i} \mathcal{L}_j [\delta_i^f(Q_i(O(t)))]\,.
\end{align}
Finally, by integrating $\frac{\dd}{\dd t} \left(e^{ 2q(\epsilon_i)^2\cosh(\beta \epsilon_i /2) \cdot t} \cdot e^{\sum_{j\neq i} t \mathcal{L}^0_j} [\delta_i^f(P_i(O(t)))]\right)$ from $0$ to $t$ (as explained in \cite{rouze2024optimalquantumalgorithmgibbs}), and taking the spectral norm, we obtain the inequality
\begin{align}
    \|\delta_i^f P_i (O(t))\| \leq\ &e^{ -2q(\epsilon_i)^2\cosh(\beta \epsilon_i /2) \cdot t} \cdot  \|\delta_i^f P_i (O(0))\|\,,
\end{align}
where we have also used that CPTP maps are contractive. Similarly also obtain
\begin{align}
    \|\delta_i^f Q_i (O(t))\| \leq\ &e^{ -q(\epsilon_i)^2\cosh(\beta \epsilon_i /2) \cdot t} \cdot  \|\delta_i^f Q_i (O(0))\|\,.
\end{align}
Adding these together, we get that
\begin{align}
     \|\delta_i^f P_i &(O(t))\| +  \|\delta_i^f Q_i (O(t))\| \leq e^{ -q(\epsilon_i)^2\cosh(\beta \epsilon_i /2) \cdot t} \cdot (\|\delta_i^f P_i (O(0))\| + \|\delta_i^f Q_i (O(0))\|)\,.
\end{align}
Denoting the unperturbed Lindbladian gap by $\Delta_0 = 2\cdot \min_i q(\epsilon_i)^2 \cosh(\beta \epsilon_i /2)$, we can sum these up in order to obtain \begin{equation}
    \oscnorm{O(t)} \leq e^{-\frac{\Delta_0}{2} t} \oscnorm{O(0)}\,,
\end{equation} like we wanted. Hence, we arrive at the following proposition:

\begin{prop}
    Free fermionic systems with bounded single particle Hamiltonian, $\|M\| = \mathcal{O}(1)$, mix rapidly in logarithmic time bounded by \begin{equation}
        t_\textup{mix} \leq \frac{2}{\Delta_0} \cdot \log\left( 8 \cdot \frac{n}{\epsilon}\right)\,.
    \end{equation}
\end{prop}

\begin{proof}
    Upper bound \begin{equation}
    \|\rho(t) - \sigma\|_{\Tr} \leq \sup_{\|O\|\leq 1} \oscnorm{O(t)} \|\rho(0) - \sigma\|_{\Tr} \leq 8n e^{-\Delta_0 t/2} \overset{\text{set}}{\leq} \epsilon
\end{equation} and solve for $t$. 
\end{proof}


\subsection{Perturbations of non-hopping fermions}

The oscillator norm technique we used to show rapid mixing of the free fermionic Lindbladian can often be extended to show rapid mixing of the perturbed Lindbladian as well \cite{rouze2024optimalquantumalgorithmgibbs, zhan2025rapidquantumgroundstate}, which follows from the locality of the resulting Lindbladian. However, in order to show the rapid mixing of the unperturbed part, we had to define the oscillator norm with respect to the Bogoliubov-transformed $b$-fermions. Since this transformation is non-local, this oscillator norm is no longer compatible with the locality of $\mathcal{L}$ which is with respect to the original $c$-fermions. There are instances where the transformation preserves locality, for example when the free Hamiltonian was already diagonal, which allows us to show rapid mixing for Hamiltonians of the form $H = \sum_{i=1}^n \epsilon_i c_i^\dagger c_i + \lambda V$, which we shall now demonstrate. Note that this case is applicable when the chemical potential of the system is the leading term in the Hamiltonian, or for example to the perturbed Kitaev chain without chemical potential, as given in \cite[Equation (4.5)]{froehlich2019lieschwingerblockdiagonalizationgappedquantum} (but note that our results are not restricted to 1D geometries).\\

The oscillator norm (Definitions \ref{def:fermionic trace}, \ref{def:fermionic osci norm}) would now be defined with respect to the original $c$-fermions. Denote the corresponding interacting and non-interacting Lindbladian by $\mathcal{L}$ and $\mathcal{L}_0$ respectively. Hence consider \begin{align}
    \frac{\dd}{\dd t} \delta_i^f(P_i(O(t))) &=  \delta_i^f(P_i(\mathcal{L}[O(t)]))\\
    &= \delta_i^f(P_i(\mathcal{L}_i^0[O(t)])) + \delta_i^f(P_i((\mathcal{L}_i - \mathcal{L}_i^0) [O(t)])) + \sum_{j \neq i} \delta_i^f(P_i(\mathcal{L}_j [O(t)]))\\ 
    = - 2q(\epsilon_i)^2\cosh(\beta \epsilon_i /2) \cdot &\delta_i^f(P_i(O(t))) + \delta_i^f(P_i((\mathcal{L}_i - \mathcal{L}_i^0) [O(t)])) + \sum_{j \neq i} \mathcal{L}_j [\delta_i^f(P_i(O(t)))] + \sum_{j \neq i} [\delta_i^fP_i,\mathcal{L}_j] [O(t)]\,,
\end{align} and similarly \begin{align}
    \frac{\dd}{\dd t} \delta_i^f(Q_i(O(t))) &= \\
     - q(\epsilon_i)^2\cosh(\beta \epsilon_i /2) \cdot &\delta_i^f(Q_i(O(t))) + \delta_i^f(Q_i((\mathcal{L}_i - \mathcal{L}_i^0) [O(t)])) + \sum_{j \neq i} \mathcal{L}_j [\delta_i^f(Q_i(O(t)))] + \sum_{j \neq i} [\delta_i^f Q_i,\mathcal{L}_j] [O(t)]\,.
\end{align}

Then again integrate $\frac{\dd}{\dd t} \left(e^{ 2q(\epsilon_i)^2\cosh(\beta \epsilon_i /2) \cdot t} \cdot e^{\sum_{j\neq i} t \mathcal{L}_j} [\delta_i^f(P_i(O(t)))]\right)$ from $0$ to $t$, and take the spectral norm to obtain the inequality \begin{align}
    \|\delta_i^f P_i (O(t))\| \leq\ &e^{ -2q(\epsilon_i)^2\cosh(\beta \epsilon_i /2) \cdot t} \cdot  \|\delta_i^f P_i (O(0))\|\\ &+\int_0^t e^{ 2q(\epsilon_i)^2\cosh(\beta \epsilon_i /2) \cdot (s-t)} \cdot \left(\left\|\delta_i^f(P_i((\mathcal{L}_i - \mathcal{L}_i^0) [O(s)]))\right\|  + \left\|\sum_{j \neq i} [\delta_i^f P_i,\mathcal{L}_j] [O(s)]\right\| \right)\ \dd s\,,
\end{align} where we have also used that CPTP maps are contractive (more specifically, their adjoints in the Heisenberg picture, $\mathcal{L}_i$, are contractive with respect to the spectral norm). Similarly also obtain \begin{align}
    \|\delta_i^f Q_i (O(t))\| \leq\ &e^{ -q(\epsilon_i)^2\cosh(\beta \epsilon_i /2) \cdot t} \cdot  \|\delta_i^f Q_i (O(0))\|\\ &+\int_0^t e^{ q(\epsilon_i)^2\cosh(\beta \epsilon_i /2) \cdot (s-t)} \cdot \left(\left\|\delta_i^f(Q_i((\mathcal{L}_i - \mathcal{L}_i^0) [O(s)]))\right\|  + \left\|\sum_{j \neq i} [\delta_i^f Q_i,\mathcal{L}_j] [O(s)]\right\| \right)\ \dd s\,.
\end{align} Adding these together, we get that \begin{align}\label{eqn: sitewise integral inequality}
     \|\delta_i^f P_i &(O(t))\| +  \|\delta_i^f Q_i (O(t))\| \leq e^{ -q(\epsilon_i)^2\cosh(\beta \epsilon_i /2) \cdot t} \cdot (\|\delta_i^f P_i (O(0))\| + \|\delta_i^f Q_i (O(0))\|)\\
     &+ \int_0^t e^{ q(\epsilon_i)^2\cosh(\beta \epsilon_i /2) \cdot (s-t)} \cdot \left(\left\|\delta_i^f(P_i((\mathcal{L}_i - \mathcal{L}_i^0) [O(s)]))\right\|  + \left\|\sum_{j \neq i} [\delta_i^f P_i,\mathcal{L}_j] [O(s)]\right\|\right.\\ &\hspace{3cm}+ \left\|\delta_i^f(Q_i((\mathcal{L}_i - \mathcal{L}_i^0) [O(s)]))\right\| + \left. \left\|\sum_{j \neq i} [\delta_i^f Q_i,\mathcal{L}_j] [O(s)]\right\| \right)\ \dd s\,.
\end{align} We would like to sum these up over the index $i$, but before that, we first make the following assumption: \begin{assumption}\label{assumption: kappa existence}
    Assume that there exist coefficients $\kappa_i^c$ and $\gamma_i^c$, which sum up to at most a constant $\eta$ like $\sum_i \kappa_i^c + \gamma _i^c \leq \eta$, and obey \begin{align}
        \|\delta_i^f(P_i((\mathcal{L}_i - \mathcal{L}_i^0) [O]))\|\ \ ,\ \ \|\delta_i^f(Q_i((\mathcal{L}_i - \mathcal{L}_i^0) [O]))\| &\leq \sum_c \kappa^c_i \cdot (\|\delta_c^f P_c(O)\| + \|\delta_c^f Q_c(O)\|)\\
        \| [\delta_i^f P_i,\sum_{j \neq i}\mathcal{L}_j] [O]\|\ \ , \ \ \| [\delta_i^f Q_i,\sum_{j \neq i}\mathcal{L}_j] [O]\| &\leq \sum_c \gamma^c_i \cdot (\|\delta_c^f P_c(O)\| + \|\delta_c^f Q_c(O)\|)
    \end{align}
\end{assumption}

Now, also denoting the unperturbed Lindbladian gap by $\Delta_0 = 2\cdot \min_i q(\epsilon_i)^2 \cosh(\beta \epsilon_i /2)$, we can sum up \eqref{eqn: sitewise integral inequality} in order to obtain \begin{equation}
    \oscnorm{O(t)} \leq e^{-\frac{\Delta_0}{2} t} \oscnorm{O(0)} + 2\eta \int_0^t e^{\frac{\Delta_0}{2} (s-t)} \oscnorm{O(s)}\ \dd s\,.
\end{equation}
This is a first order integral inequality, known as Gr\"onwall's inequality, which we can solve when $\eta \leq \Delta_0 / 4$ by
\begin{equation}
    \oscnorm{O(t)} \leq e^{-(\Delta_0 /2 - 2\eta)t} \cdot \oscnorm{O(0)}\,.
\end{equation}
Hence, we arrive at the following theorem:

\begin{theorem}
    For quasi-local fermionic systems governed by the Hamiltonian $H = \sum_{i=1}^n \epsilon_i c^\dagger_i c_i + \lambda V$, there exists a constant $\lambda_\textup{max}$ such that for $|\lambda| \leq \lambda_\textup{max}$, the Lindbladian mixes rapidly in logarithmic time bounded by \begin{equation}
        t_\textup{mix} \leq \frac{1}{\Delta_0 /2 - 2\cdot \eta(\lambda)} \cdot \log\left( 8 \cdot \frac{n}{\epsilon}\right)\,.
    \end{equation}
\end{theorem}

This theorem follows from upper bounding \begin{equation}
    \|\rho(t) - \sigma\|_{\Tr} \leq \sup_{\|O\|\leq 1} \oscnorm{O(t)} \|\rho(0) - \sigma\|_{\Tr} \leq 8n e^{-(\Delta_0 /2 -2\eta)t} \overset{\text{set}}{\leq} \epsilon
\end{equation} and solving for $t$. Finally, to show the existence of $\eta$, we refer back to Lemma \ref{lemma:eta}. Note that understanding $\eta(\lambda)$ then also yields and explicit value for $\lambda_\textup{max}$ as in Corollary \ref{cor:maximal lambda}.

\subsection{Strongly-interacting Fermi-Hubbard model}

In this section, we show how to adapt the techniques for separable Hamiltonians from Section \ref{sec:Spins} to the fermionic setting, and in particular we show rapid mixing of the Fermi-Hubbard model in the strongly-interacting limit. The full Hamiltonian is given by \begin{equation}
    H_\text{FH} = -t\sum_{\langle i,j\rangle, \sigma} \left( c^\dagger _{i,\sigma} c_{j,\sigma} + c^\dagger_{j,\sigma}c_{i,\sigma} \right) + U\sum_{i=1}^n N_{i,\uparrow} N_{i,\downarrow}\,,
\end{equation} where $\langle \cdot,\cdot\rangle$ denotes nearest neighbours on the lattice, $\sigma \in \{\uparrow,\downarrow\}$ represents the two different spins, and $N_{i,\sigma} = c^\dagger _{i,\sigma} c_{i,\sigma}$ is the fermionic number operator of the given spin. We can observe that the interaction part of the system is indeed separable in lattice sites, and so we consider \begin{equation}
    H_0 = U\sum_{i=1}^n N_{i,\uparrow} N_{i,\downarrow}\label{eqn: Hamiltonian, unperturbed strong FH}
\end{equation} as our base Hamiltonian, while the hopping part will be considered as its perturbation.

Now we select the set of jump operators to be $\mathcal{A} = \{c_{i,\downarrow},c^\dagger_{i,\downarrow},c_{i,\uparrow},c^\dagger_{i,\uparrow}\}_{i=1}^n$, i.e.~all the individual ladder operators, and evaluate the corresponding Lindblad operators. We can readily find the time evolution of the jump operators as \begin{align}
    e^{iH_0 t} c_{i,\sigma} e^{-iH_0 t} &= c_{i,\sigma} \cdot \left(1-N_{i,\overline{\sigma}} +e^{-iUt}N_{i,\overline{\sigma}} \right)\,,\\
    e^{iH_0 t} c^\dagger_{i,\sigma} e^{-iH_0 t} &= c^\dagger_{i,\sigma} \cdot \left(1-N_{i,\overline{\sigma}} +e^{+iUt}N_{i,\overline{\sigma}}\right)\,,
\end{align} where $\overline{\sigma}$ denotes the opposite spin to $\sigma$. Hence we find that \begin{align}
    L_i^{(1,\sigma)} &= c_{i,\sigma} \cdot \left(\hat{f}_i(0)\cdot(1-N_{i,\overline{\sigma}} )+\hat{f}_i(-U)\cdot N_{i,\overline{\sigma}} \right)\,,\\
    L_i^{(2,\sigma)} &= c^\dagger_{i,\sigma} \cdot \left(\hat{f}_i(0)\cdot(1-N_{i,\overline{\sigma}} )+\hat{f}_i(+U)\cdot N_{i,\overline{\sigma}} \right)\,.
\end{align} Now let's assume for simplicity that the filter functions are all the same, $\hat{f}_i(\nu) = \hat{f}(\nu)$, are real in the frequency domain, and are also normalised such that $\hat{f}(0)=1$. Denoting $\hat{f}(\pm U) = f_\pm$, we can evaluate the sum \begin{align}
    \sum_{t,\sigma} L_i^{\dagger\ (t,\sigma)}  L_i^{(t,\sigma)} = 2 + (f_+^2 - 1)\cdot (N_{i,\downarrow} + N_{i,\uparrow}) + 2 (f_-^2 - f_+ ^2) \cdot N_{i,\uparrow}N_{i,\downarrow}\,,
\end{align} and notice that it commutes with the Hamiltonian, $\left[H_0, \sum_{t,\sigma} L_i^{\dagger\ (t,\sigma)}  L_i^{(t,\sigma)}\right] = 0$, which ensures that the individual coherent terms $G_i$ vanish. Finally, we obtain the action of the Lindbladian in the Heisenberg picture as \begin{align}
    \mathcal{L}^0_i[O] = &\sum_\sigma \left( \left[1-N_{i,\overline{\sigma}}+f_- \cdot N_{i,\overline{\sigma}}\right] \cdot c^\dagger _{i,\sigma} \cdot O \cdot c_{i,\sigma} \cdot \left[1-N_{i,\overline{\sigma}}+f_- \cdot N_{i,\overline{\sigma}}\right]\right. \\
    &\qquad+ \left.\left[1-N_{i,\overline{\sigma}}+f_+ \cdot N_{i,\overline{\sigma}}\right] \cdot c_{i,\sigma} \cdot O \cdot c^\dagger_{i,\sigma} \cdot \left[1-N_{i,\overline{\sigma}}+f_+ \cdot N_{i,\overline{\sigma}}\right] \right)\\
    &-\frac{1}{2} \left\{ 2 + (f_+^2 - 1)\cdot (N_{i,\downarrow} + N_{i,\uparrow}) + 2 (f_-^2 - f_+ ^2) \cdot N_{i,\uparrow}N_{i,\downarrow}\ ,\  O \right\}\,.
\end{align}

In order to diagonalise this Lindbladian, we shall proceed carefully in order to account for the subtleties due to the fermionic statistics. Firstly observe that the parts of the Lindbladian acting on different sites still commute with each other, $\mathcal{L}^0_i \circ \mathcal{L}^0_j = \mathcal{L}^0_j \circ \mathcal{L}^0_i$. Now consider a basis for observables $O$, where we separate the contribution at the site $i$, $O=F_i \otimes O_i$, with $F_i$ supported only at site $i$ and $O_i$ supported outside of site $i$. In the end, we will need to consider only observables that are in total even, but the situation will differ depending on whether the contribution at site $i$ is even or odd. Let's start by considering $O = F_i \otimes O_i^{(\text{even})}$ (where $F_i$ can now be either even or odd). In this case, we simply have that $\mathcal{L}_i[O] = \mathcal{L}_i[F_i] \otimes O_i^{(\text{even})}$, and we can proceed with the diagonalisation as usual. However, for $O = F_i \otimes O_i^{(\text{odd})}$, we find that \begin{align}
    \mathcal{L}^0_i[O] = \bigg(&\sum_\sigma \left( -\left[1-N_{i,\overline{\sigma}}+f_- \cdot N_{i,\overline{\sigma}}\right] \cdot c^\dagger _{i,\sigma} \cdot F_i \cdot c_{i,\sigma} \cdot \left[1-N_{i,\overline{\sigma}}+f_- \cdot N_{i,\overline{\sigma}}\right]\right. \\
    &\qquad- \left.\left[1-N_{i,\overline{\sigma}}+f_+ \cdot N_{i,\overline{\sigma}}\right] \cdot c_{i,\sigma} \cdot F_i \cdot c^\dagger_{i,\sigma} \cdot \left[1-N_{i,\overline{\sigma}}+f_+ \cdot N_{i,\overline{\sigma}}\right] \right)\\
    &-\frac{1}{2} \left\{ 2 + (f_+^2 - 1)\cdot (N_{i,\downarrow} + N_{i,\uparrow}) + 2 (f_-^2 - f_+ ^2) \cdot N_{i,\uparrow}N_{i,\downarrow}\ ,\  F_i \right\}\bigg) \otimes O_i^{(\text{odd})}\\
    &\equiv \mathcal{L}_i^{(\text{twisted})}[F_i] \otimes O_i^{(\text{odd})}\,,
\end{align} where we have introduced minus signs in front of the first part of the Lindbladian, and denoted this new version as "twisted". In order to build up all the eigenvectors of the full Lindbladian, we would need to consider both $F_i^{(\text{twisted,even})}$ and $F_i^{(\text{twisted,odd})}$, however, since we only care about observables which are in total even, we only need eigenvectors of the forms $F_i^{(\text{even})} \otimes O_i^{(\text{even})}$ and $F_i^{(\text{twisted,odd})} \otimes O_i^{(\text{odd})}$. 

We can then find that the $8$ eigenvectors $F_i^{(\text{twisted,odd})}$ will be of the forms
\begin{align}
    F_i^{\pm,\sigma} &\propto \left(1-2f_+-f_-^2\pm \sqrt{4f_-f_++f_-^4-2f_-^2+1}\right) \cdot c_{i,\sigma} + 2(f_+-f_-+f_-^2-1)\cdot N_{i,\overline{\sigma}}c_{i,\sigma}\\
    F_i^{\pm,\sigma,\dagger} &\propto \left(1-2f_+-f_-^2\pm \sqrt{4f_-f_++f_-^4-2f_-^2+1}\right) \cdot c^\dagger_{i,\sigma} + 2(f_+-f_-+f_-^2-1)\cdot N_{i,\overline{\sigma}}c^\dagger_{i,\sigma}
\end{align} and have eigenvalues \begin{align}
    \lambda_\pm^{(\text{odd})} = \frac{1}{2}\left(-f_-^2-f_+^2-2\mp \sqrt{4f_-f_++f_-^4-2f_-^2+1}\right)\,.
\end{align} Note that the difference between $F_i^{(\text{odd})}$ and $F_i^{(\text{twisted,odd})}$ amounts only into mapping $f_\pm \mapsto -f_\pm$, and so their eigenvalues remain the same, while the coefficients of the eigenvectors undergo these sign changes. For the even eigenvectors $F_i^{(\text{even})}$, we find that \begin{align}
    F_i^{(9)} \propto c_i^\uparrow c_i^\downarrow, \quad F_i^{(10)} \propto c_i^{\dagger \uparrow} c_i^{\dagger \downarrow}
\end{align} are eigenvectors with eigenvalues \begin{align}
    \lambda^{(\text{even})}_{9,10} = -1-f_-^2\,,
\end{align} then \begin{align}
    F_i^{(11)} \propto c_i^{\dagger \uparrow} c_i^\downarrow, \quad F_i^{(12)} \propto c_i^{\uparrow} c_i^{\dagger \downarrow}
\end{align} are eigenvectors with eigenvalues \begin{align}
    \lambda^{(\text{even})}_{11,12} = -1-f_+^2\,,
\end{align} and finally we have \begin{align}
    F_i^{(13)} &\propto I\quad \text{with}\quad \lambda_{13}^{(\text{even})}=0\\
    F_i^{(14)} &\propto N_{i,\uparrow} - N_{i,\downarrow} \quad \text{with}\quad \lambda_{14}^{(\text{even})}=-1-f_+^2\\
    \lambda^{(\text{even})}_{15,16} &= \frac{1}{2}\left( -f_+^2-2f_-^2-3 \mp \sqrt{f_+^4+4f_+^2f_-^2 - 2f_+^2+4f_-^4-12f_-^2+9} \right) 
\end{align} with $F_i^{(15,16)}$ being given in Appendix \ref{appendix: eigenvectors} due to their length.
In particular, we see that the identity is a unique $0$-eigenstate and the other eigenvalues are negative (as expected), and that the spectral gap of $\mathcal{L}^0$ is \begin{align}
    \Delta_0 = \min\bigg\{ 1+f_-^2\,,\,1+f_+^2\,,\,&\frac{1}{2}\left(f_-^2+f_+^2+2- \sqrt{4f_-f_++f_-^4-2f_-^2+1}\right)\,,\\ &\frac{1}{2}\left( f_+^2+2f_-^2+3 - \sqrt{f_+^4+4f_+^2f_-^2 - 2f_+^2+4f_-^4-12f_-^2+9} \right)\bigg\}\,.\label{eqn:spectral gap unperturbed strong FH}
\end{align}

\begin{figure}[h]
    \centering
    \includegraphics[width=0.7\linewidth]{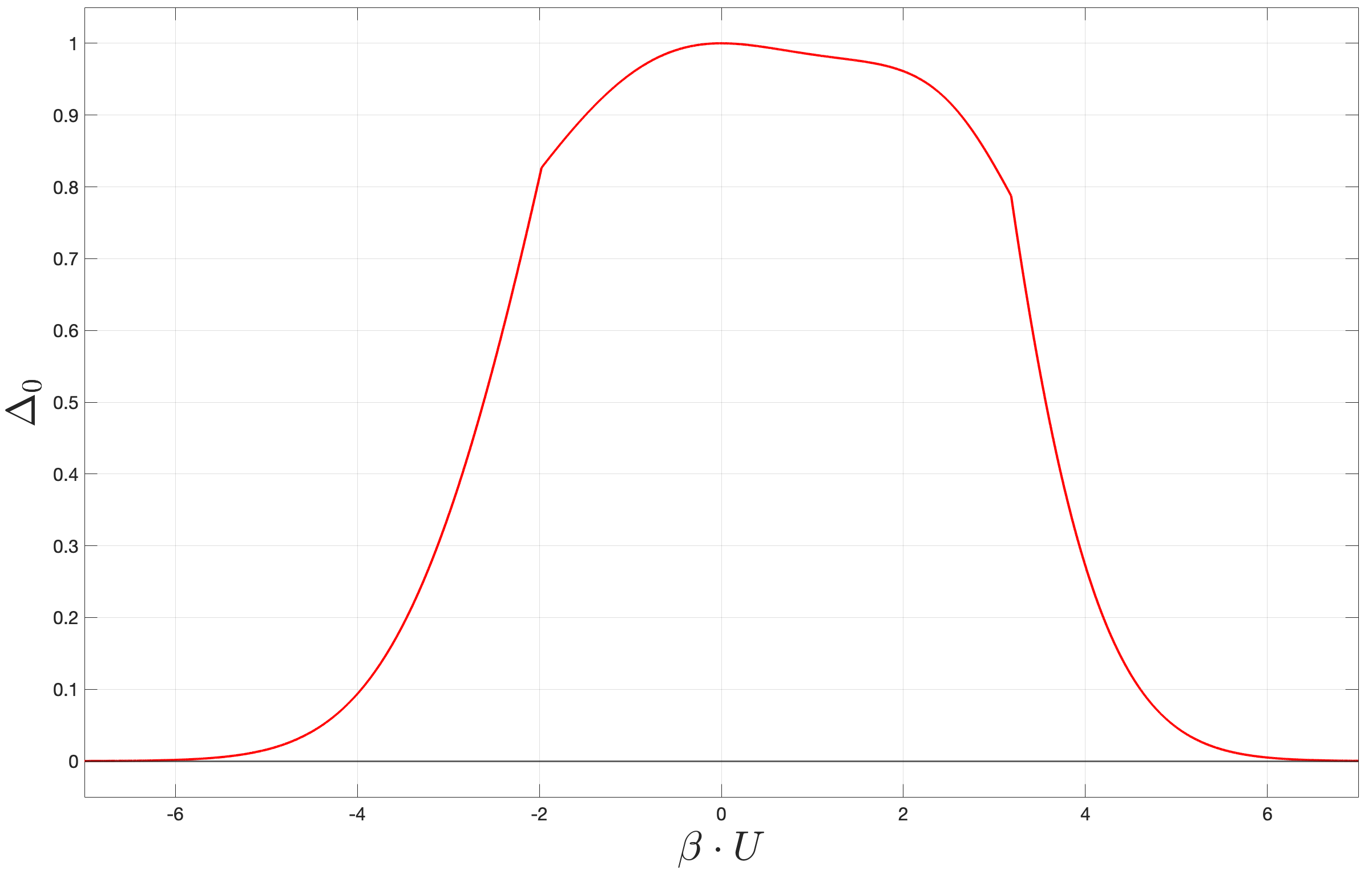}
    \caption{Plot of the unperturbed spectral gap $\Delta_0$ from Equation \eqref{eqn:spectral gap unperturbed strong FH} of the Lindbladian $\mathcal{L}_0$ corresponding to the strongly-interacting limit of the Fermi-Hubbard model given by $H_0$ from Equation \eqref{eqn: Hamiltonian, unperturbed strong FH}, when using the Gaussian filter function $\hat{f}(\nu) = e^{-\beta^2 \nu^2 /8 - \beta\nu/4}$.}
    \label{fig:unperturbed_gap_strong_FH}
\end{figure}

Next, we need to define the partial quasi-derivations with respect to these $16$ eigenvectors (which we shall number from $1$ to $16$ in the order of decreasing eigenvalues, starting with $\lambda_i^1=0$), like \begin{align}
    \delta_i^\alpha(O) = \delta_i(F_i^\alpha)\cdot \langle F_i^\alpha,O\rangle_{\sigma_i^0}\,,
\end{align} where $\delta_i(O) = O - \frac{1}{4}\Tr_i^f(O)$ and $\langle F_i^\alpha,O\rangle_{\sigma_i^0} = \Tr_i^f \left(F_i^\alpha \cdot e^{-\frac{\beta}{2} U N_{i,\uparrow}N_{i,\downarrow}} \cdot O\cdot e^{-\frac{\beta}{2} U N_{i,\uparrow}N_{i,\downarrow}}\right)$, with the fermionic partial trace given by $\Tr_i^f = \Tr_i^{(c_\downarrow)} \circ \Tr_i^{(c_\uparrow)} =\Tr_i^{(c_\uparrow)} \circ \Tr_i^{(c_\downarrow)}$ and $\Tr_i^{(c)}$ from Equation \eqref{eqn:fermionic trace}. Note that the eigenvectors $F_i^\alpha$ are orthogonal with respect to the KMS inner product, $\langle F_i^\alpha,F_i^\beta\rangle_{\sigma_i^0} = \delta_{\alpha\beta}$, where we have further implicitly assumed that they are also normalised. In other words, we can expand the observable $O$ along the site $i$ in the eigenbasis $\{F_i^\alpha\}_{\alpha=1}^{16}$, and then $\delta_i^\alpha(O)$ takes the full quasi-derivation $\delta_i$ of the component along $F_i^\alpha$. This ensures the key relation \begin{align}
    \delta_i^\alpha(\mathcal{L}_i^0[O]) = \lambda^\alpha \cdot \delta_i^\alpha(O)\,.
\end{align} Further, neither the partial quasi-derivations nor the Lindbladian will change the (even) parity of the observable $O$, and so we get that $\delta_i^\alpha \circ \mathcal{L}^0_j [O] =  \mathcal{L}^0_j  \circ \delta_i^\alpha[O]$ for all $i \neq j$, and hence we can proceed as before to get rapid mixing of $\mathcal{L}_0$ in time \begin{align}
    t_\textup{mix}(\epsilon) \leq \frac{1}{\Delta_0} \log\left( 1024\cdot e^{2\beta |U|} \frac{n}{\epsilon} \right)\,.
\end{align}
Finally, since the partial quasi-derivations $\delta_i^\alpha$ are compatible with the locality of the system, we can extend the rapid mixing result to the regime of small enough $|t| \leq t_\text{max}$ to arrive at the following proposition:
\begin{prop}
    For the Fermi-Hubbard model given by $
    H_\textup{FH} = -t\sum_{\langle i,j\rangle, \sigma} \left( c^\dagger _{i,\sigma} c_{j,\sigma} + c^\dagger_{j,\sigma}c_{i,\sigma} \right) + U\sum_{i=1}^n N_{i,\uparrow} N_{i,\downarrow}$, where $\langle i,j\rangle$ represents nearest neighbours on the lattice, $\sigma \in \{\uparrow,\downarrow\}$ represents the spin of the fermions, and $N_{i,\sigma} = c^\dagger _{i,\sigma} c_{i,\sigma}$, there exists a system-size-independent maximal hopping strength $t_\textup{max}$ (which depends on the fixed interaction strength $U$ and temperature $\beta$) such that for $|t|\leq t_\textup{max}$, the corresponding Lindbladian $\mathcal{L}$ mixes rapidly in time bounded by 
    \begin{align}
    t_\textup{mix}(\epsilon) \leq \frac{1}{\Delta_0-16\eta(t)} \log\left( 1024\cdot e^{2\beta |U|} \frac{n}{\epsilon} \right)\,,
    \end{align}
    where $\eta(t)$ is given by Lemma \ref{lemma:eta}.
    Further, the constant $t_\textup{max}$ can be explicitly evaluated using Corollary \ref{cor:maximal lambda}.
\end{prop}

\subsubsection{Evaluating $t_\textup{max}$}
\label{sec: Evaluating explicit parameters for FH model}

We now wish to explicitly certify parameter regimes of the Fermi-Hubbard model for which our theory guarantees rapid mixing by evaluating the maximal hopping strength $t_\text{max}$ using Corollary \ref{cor:maximal lambda}, focusing on the 2D square lattice. This will turn out to be a bit more complicated than just a direct evaluation due to a few reasons. Firstly, as explained at the very end of Section \ref{sec: rapid mixing of qudits}, we need to find a sufficiently large radius $r_0$ to guarantee $t_\text{max}$ to be positive, and so we need to carry out an optimisation over $r_0$.
Next, as derived in Appendix \ref{sec: Localising odd fermionic operators}, we have the following Lieb-Robinson bound on locality for the Fermi-Hubbard model on a $D$-dimensional cubic lattice:
\begin{equation}
    \forall_{\mu>0}: \|c_i^{(\dagger)}(\tau)-c_i^{(\dagger)(r)}(\tau)\| \leq 2 e^{-\mu r}\left( e^{2(|U|+8D|t|e^{\mu})|\tau|}-1\right)\,.
\end{equation} Here, $\mu$ is a free parameter, which we also need to optimise in order to maximise $t_\text{max}$. Finally, since the Lieb-Robinson bounds depend on the hopping strength $t$, but the maximal allowed hopping strength $t_\text{max}$ depends on the Lieb-Robinson bounds, we end up with a highly non-linear relation for $t_\text{max}$.

In summary, in order to find the maximal guaranteed hopping strength $t_\textup{max}$ covered by our proof, we need to solve the equation \begin{align}
    t = \max\limits_{\overset{r_0 \in \mathbb{N}}{\mu>0}} \lambda_\textup{max}(t,r_0,\mu)\,,
\end{align} where the functional form of $\lambda_\textup{max}$ is specified in Corollary \ref{cor:maximal lambda}. We carry out this calculation numerically and present the results on Figure \ref{fig: certified regimes for FH}, displaying the certified regimes of rapid mixing in red. Note that, while this plot seems fairly symmetric, it is not even in $\beta \cdot U$, as the unperturbed spectral gap $\Delta_0$ isn't. Further, for the certified regimes, we plot the prefactor $\frac{1}{\Delta_0-16\eta(t)}$ of our mixing time upper bound on Figure \ref{fig: mixing times prefactor for FH}. It is also important to note that these results on guaranteed rapid mixing do not imply that the system won't be rapidly mixing beyond the shown regimes, nor that the values of $t_\text{max}$ obtained below for which rapid mixing is guaranteed cannot be improved by, e.g., using tighter Lieb-Robinson bounds obtained with the technique from \cite{Wang2020} for the strongly-interacting Fermi-Hubbard model, or tighter estimates on $\|\mathcal{L}_i - \mathcal{L}_i^0\|_{\infty \to \infty}$.

\begin{figure}[H]
    \centering
    \includegraphics[width=\linewidth]{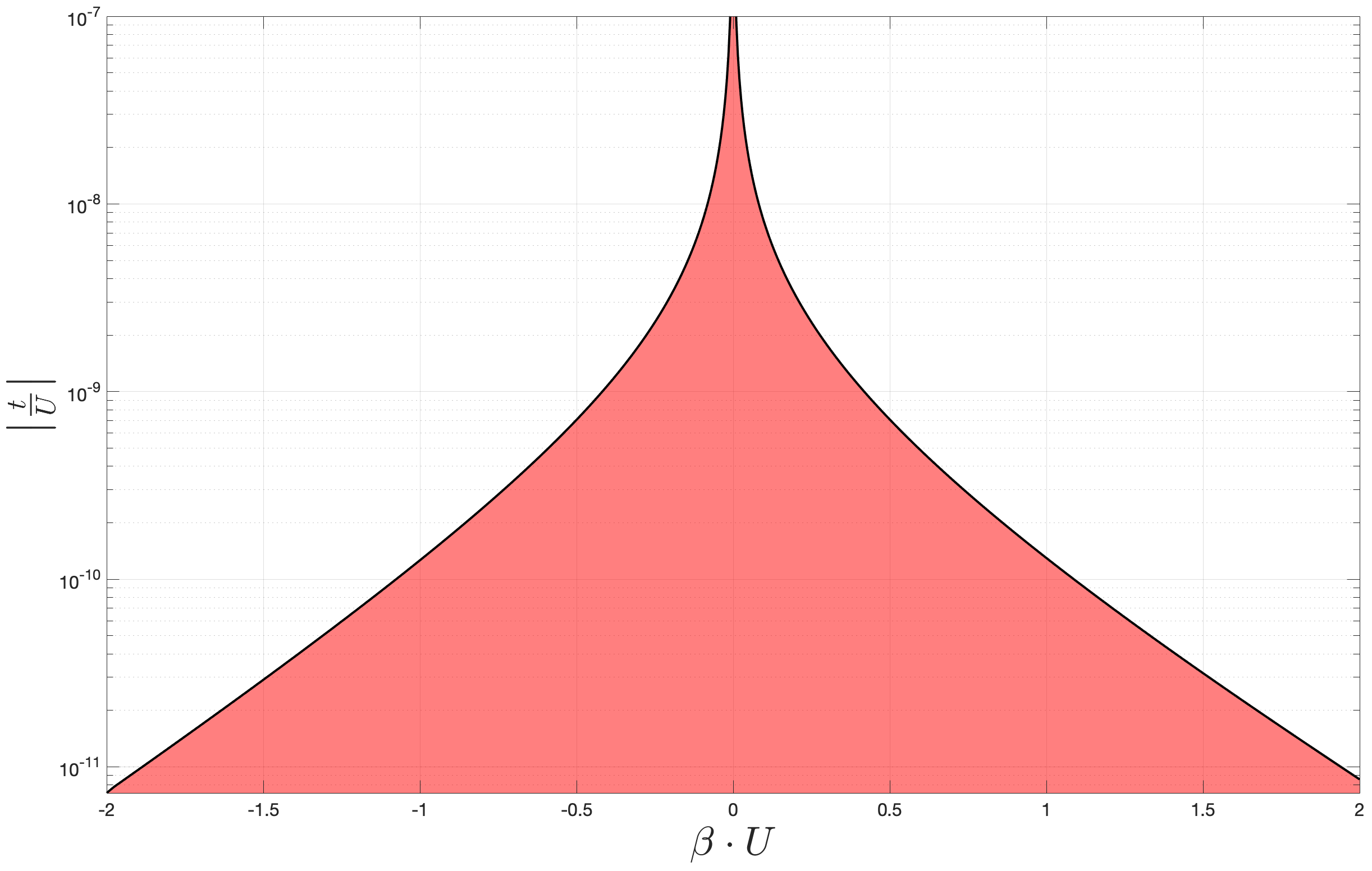}
    \caption{Parameter regimes for which the rapid mixing of the Fermi-Hubbard model is guaranteed, corresponding to the strongly-interacting regime on a 2D square lattice, using the Gaussian filter function.}
    \label{fig: certified regimes for FH}
\end{figure}

In order to evaluate $\|\mathcal{L}_i - \mathcal{L}_i^0\|_{\infty \to \infty}$, we needed to express the perturbation of the Hamiltonian as a sum over terms supported on different balls of different radii $r$, i.e. \begin{align}
    \sum_{\langle i,j\rangle, \sigma} \left( c^\dagger _{i,\sigma} c_{j,\sigma} + c^\dagger_{j,\sigma}c_{i,\sigma} \right) \overset{\text{set}}{=} \sum_{r\ge 1} \sum_{C \in C(r)} W_C\,,
\end{align} such that $\|W_C\| \leq Ke^{-\nu r}$ (see Lemma \ref{lemma:bounds on delta L}). Since the Fermi-Hubbard model is a nearest-neighbour system, we only require balls of radius $1$. Then, by always choosing the bottom and left edges of every square in the grid, i.e.~choosing the set $\{W_{i,j}\}_{i,j}$ with \begin{align}
    W_{i,j} = \sum_{\sigma} \left( c^\dagger _{i,j,\sigma} c_{i+1,j,\sigma} + c^\dagger_{i+1,j,\sigma}c_{i,j,\sigma} + c^\dagger _{i,j,\sigma} c_{i,j+1,\sigma} + c^\dagger_{i,j+1,\sigma}c_{i,j,\sigma} \right)\,,
\end{align} (and appropriate simplifications along the very top and right edges of the system), we would require \begin{align}
    4 \leq Ke^{-\nu}\,,
\end{align} where we again have one extra degree of freedom due to the short-range nature of the system. The bound on $\|\mathcal{L}_i - \mathcal{L}_i^0\|_{\infty \to \infty}$ subsequently depends on \begin{align}
    K\cdot \Phi_2(e^{-\nu}) = K\cdot \sum_{r\geq 1} (2r+1)^2 e^{-\nu r} = K\cdot \frac{9e^{2\nu} - 2e^\nu + 1}{(e^\nu - 1)^3}\,.
\end{align}
We wish to minimise $\|\mathcal{L}_i - \mathcal{L}_i^0\|_{\infty \to \infty}$, which corresponds to minimising $K\cdot \Phi_2(e^{-\nu})$, and since \begin{align}
    K\cdot \Phi_2(e^{-\nu}) \geq 4 e^\nu \cdot \frac{9e^{2\nu} - 2e^\nu + 1}{(e^\nu - 1)^3} \geq 36\,,
\end{align} which is obtained for $K = 4e^\nu$ and $\nu \to \infty$, we will use the value $36$ for this term.

\begin{figure}[H]
    \centering
    \includegraphics[width=\linewidth]{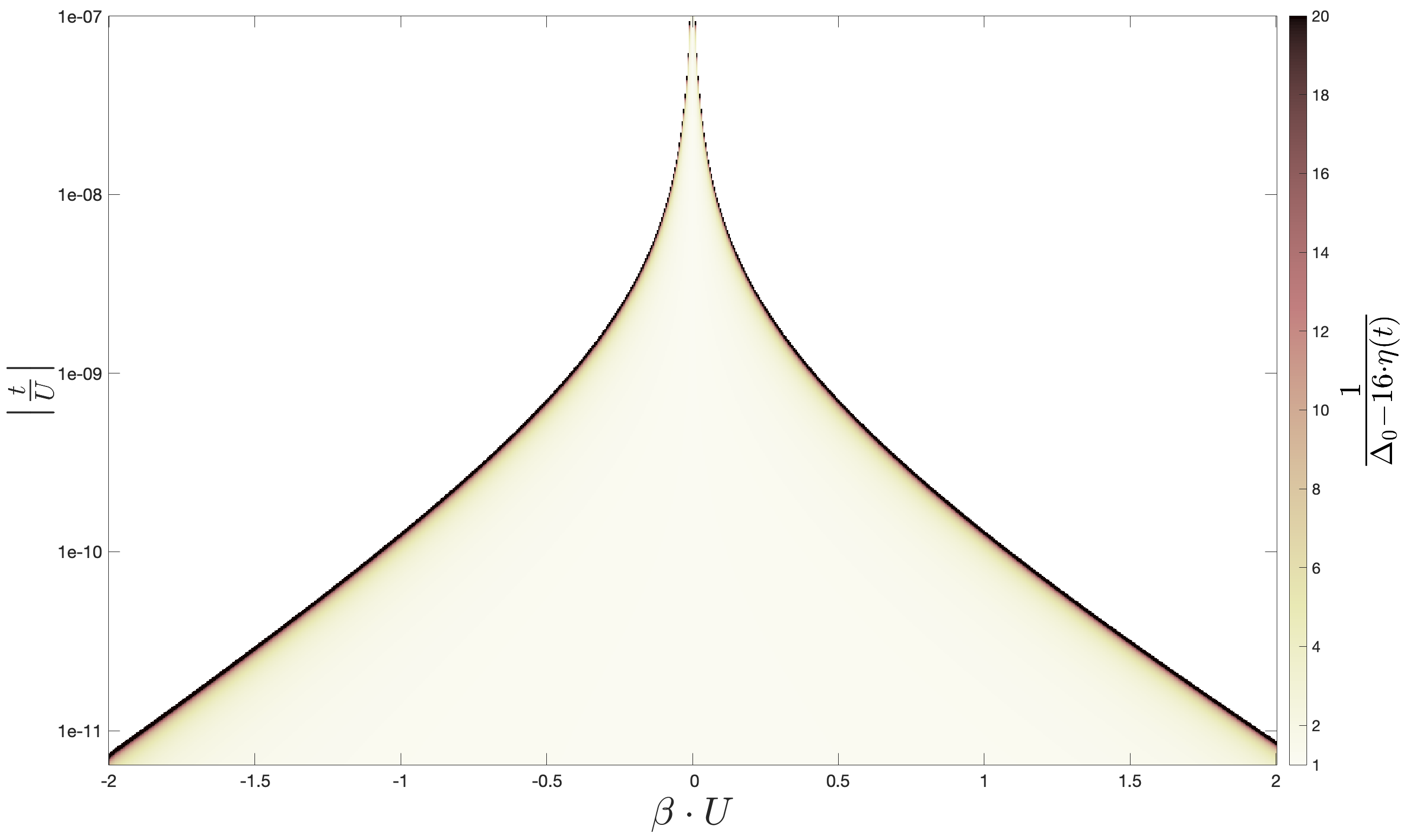}
    \caption{Parameter regimes with guaranteed rapid mixing as in Figure \ref{fig: certified regimes for FH}, but plotting the constant prefactor $\frac{1}{\Delta_0 - 16\eta(t)}$ of the mixing time $t_\textup{mix}(\epsilon)$ as a heatmap. The darker the shade, the slower the mixing time will be, with the prefactor diverging to $\infty$ at the border of guaranteed rapid mixing.}
    \label{fig: mixing times prefactor for FH}
\end{figure}

\newpage
\clearpage


\section{Bosonic Systems}\label{sec:bosons}

Here, we shall consider a free bosonic Hamiltonian of the form
\begin{align}
    H_0 = \sum_{i,j=1}^n a_i^\dagger h_{ij}a_j = \mathbf{a}^\dagger \cdot h \cdot \mathbf{a}
\end{align}
with $h$ hermitian, and $a_i,a_i^\dagger$ obeying the canonical commutation relations. We shall further assume $h$ to be real for future convenience, and so it will be a symmetric matrix. Since the bosonic ladder operators are unbounded, the oscillator norm technique is immediately not applicable, and so we shall proceed by analysing the behaviour for Gaussian states, which will impose the restriction on the initial state of the evolution.

\begin{prop}
    The free bosonic system $H_0$ has a well-defined Gibbs state if and only if $h$ is positive definite.
\end{prop}
\begin{proof}
    The partition function of the system is
    \begin{equation}
        Z = \Tr(e^{-\beta H_0}) = \Tr(e^{-\beta \sum_i \epsilon_i b_i^\dagger b_i}) = \prod_i \Tr(e^{-\beta \epsilon_i b_i^\dagger b_i}) =  \prod_i \sum_{n^{(i)} \geq 0} e^{-\beta \epsilon_i n^{(i)}} = \prod_i \frac{1}{1-e^{-\beta \epsilon_i}}\,,
    \end{equation} where the sum converges if and only if $\epsilon_i >0$ for all $i$.
\end{proof}

This condition will be important later for the Lindbladian to even have a steady state. Unlike Lindbladians over finite-dimensional Hilbert spaces, which must have a non-positive eigenspectrum and at least one steady state, bosonic Lindbladians can have positive eigenvalues and no steady states. From now on, we shall assume $h$ to be positive definite unless stated otherwise.

\begin{prop}
    For a free bosonic system $H_0$, by taking the set of jump operators to be $\{x_i = a_i + a_i ^\dagger, p_i =-i(a_i-a_i^\dagger)\}_{i=1}^n$, and filter functions $\hat f_a$ to be real and equal, the coherent term $G$ in the algorithmic Lindbladian vanishes.
\end{prop}
\begin{proof}
    We find that $[H_0,a_k] = -\sum_jh_{kj}a_j$ and $[H_0,a_k^\dagger] = \sum_jh_{kj}a_j^\dagger$, hence \begin{align}
        \mathbf{a}(t) &= e^{it\operatorname{ad}_{H_0}}[\mathbf{a}] = e^{-ith}\cdot\mathbf{a}\,,\\
        \mathbf{a}^\dagger(t) &= e^{it\operatorname{ad}_{H_0}}[\mathbf{a}^\dagger] = e^{ith}\cdot\mathbf{a^\dagger}\,.
    \end{align} The time evolved jump operators are then \begin{align}
        \mathbf{x}(t) &= e^{-ith} \cdot \mathbf{a} + e^{ith} \cdot \mathbf{a}^\dagger\,,\\
        \mathbf{p}(t) &= -i (e^{-iht} \cdot \mathbf{a} - e^{iht} \cdot \mathbf{a}^\dagger)\,,
    \end{align} which corresponds to the Lindblad operators being \begin{align}
        \mathbf{L}_x &= \hat{f}(-h) \cdot \mathbf{a} + \hat{f}(h) \cdot \mathbf{a}^\dagger\,,\\
        \mathbf{L}_p &= -i \hat{f}(-h) \cdot \mathbf{a} + i \cdot \hat{f}(h) \cdot \mathbf{a}^\dagger\,.
    \end{align}
    Finally, we find that \begin{align}
        \sum_\mu L_\mu^\dagger L_\mu &= \mathbf{L}_x^\dagger \cdot \mathbf{L}_x + \mathbf{L}_p ^
        \dagger \cdot \mathbf{L}_p \\
        &= 2 \mathbf{a}^\dagger \cdot (\hat{f}(-h)^2 + \hat{f}(h)^2) \cdot \mathbf{a} + 2\Tr(\hat{f}(h)^2)\,,
    \end{align} but then -- observing that $[H_0, \mathbf{a}^\dagger \cdot S \cdot \mathbf{a}] = \mathbf{a}^\dagger \cdot [h,S] \cdot \mathbf{a}$ -- this quantity will commute with $H_0$, and so we obtain \begin{equation}
        G = \int_{-\infty}^\infty g(t) \cdot e^{iH_0 t}  \sum_\mu L_\mu^\dagger L_\mu e^{-iH_0 t}\ \dd t \propto \hat{g}(0) \propto \tanh(0) = 0\,,
    \end{equation} showing that the coherent term vanishes.
\end{proof}

\begin{prop}
    The Lindbladian $\mathcal{L}_0^\dagger$ corresponding to the free bosonic Hamiltonian $H_0$ with the set of jump operators $\{x_i,p_i\}_{i=1}^n$ and equal real filter functions $\hat{f}^a(\nu) = \hat{f}(\nu) = q(\nu) e^{-\beta \nu/4}$ has spectral gap given by \begin{equation}
        \Delta_0 = 2 \cdot \min_i q(\epsilon_i)^2 \cdot \sinh\left(\frac{\beta}{2} \epsilon_i\right)\,,
    \end{equation} where $\epsilon_i \in \operatorname{spec}(h)$ are the eigenvalues of the single particle Hamiltonian $h$.
\end{prop}
\begin{proof}
    Here we shall closely follow third quantisation from \cite{prosen2010quantizationbosonoperatorspaces}. From \cite[Equation (11)]{prosen2010quantizationbosonoperatorspaces}, we get $H = 0 = K$; while from \cite[Equations (12,15)]{prosen2010quantizationbosonoperatorspaces}, we get that $M = \hat{f}(-h)^2$, $N = \hat{f}(h)^2$, and $L = 0$. This gives the effective Hamiltonian \cite[Equation (18)]{prosen2010quantizationbosonoperatorspaces} as \begin{equation}
        X = \frac{1}{2} \left( \begin{matrix}
            -\hat{f}(h)^2 + \hat{f}(-h)^2 & 0 \\ 0 & -\hat{f}(h)^2 + \hat{f}(-h)^2
        \end{matrix} \right) =  I_2 \otimes q(h)^2 \cdot 
            \sinh\left(\frac{\beta}{2}h\right)
    \end{equation} 
    Hence the rapidities, i.e. the eigenvalues of $X$, are $\beta_{i,\pm} = q(\epsilon_i)^2 \sinh\left(\frac{\beta}{2}\epsilon_i\right)$ (here the $\pm$ represents that each of them is present twice). Hence we find the full spectrum of the Lindbladian as \begin{equation}
        \operatorname{spec}(\mathcal{L}_0^\dagger ) = \left\{ -2 \sum_{i=1}^n \beta_i \cdot (n^+_i + n^-_i) \right\}_{n_i^\pm \in \mathbb{N}_0}\,,
    \end{equation} and specifically we find the gap to be \begin{equation}
        \Delta_0 = 2 \cdot \min_i q(\epsilon_i)^2 \cdot \sinh\left(\frac{\beta}{2} \epsilon_i\right)\,.
    \end{equation}
\end{proof}

\begin{corollary}
    For free bosonic Hamiltonians, which have a bounded single particle Hamiltonian -- meaning $\|h^{-1}\| \leq \mathcal{O}(1)$ and $\|h\| \leq \mathcal{O}(1)$ -- the Lindbladian $\mathcal{L}_0^\dagger$ has a constant spectral gap $\Delta_0$.
\end{corollary}
\begin{proof}
    Using the Gaussian filter function with $q(\nu) = e^{-\beta^2 \nu^2 /8}$, the spectral gap of the Lindbladian becomes \begin{align}
        \Delta_0 = \min\left\{2e^{-\|h^{-1}\|^{-2} /4} \sinh\left(\frac{\beta}{2}\|h^{-1}\|^{-1}\right),2e^{-\|h\|^2 /4} \sinh\left(\frac{\beta}{2}\|h\|\right)\right\}\,,
    \end{align} which is indeed lower bounded by a constant from the assumptions on $h$.
\end{proof}

\begin{remark}
    Unlike in the case of fermions, the chemical potential of the system, which is included in $h$, can qualitatively change the behaviour of the convergence. For chemical potentials such that $\exists_{\epsilon>0}:\forall_{n\in \mathbb{N}}: \epsilon_\text{min} > \epsilon$, where $\epsilon_\text{min}$ is the lowest eigenvalue of $h$, the spectral gap is constant. For chemical potentials such that $\epsilon_\text{min} < 0$, the Gibbs state is not well defined. Then, however, there exists also a critical value of the chemical potential, where $\epsilon_\text{min}>0$ for all system sizes $n$, however $\epsilon_\text{min} \to 0$ as $n\to \infty$. In this case, each system size $n$ has a well defined Gibbs state, but we can see the spectral gap of the Lindbladian closing polynomially with respect to $n$.
\end{remark}

\begin{prop}
     For free bosonic Hamiltonians, which have a bounded single particle Hamiltonian -- meaning $\|h^{-1}\| \leq \mathcal{O}(1)$ and $\|h\| \leq \mathcal{O}(1)$ -- when taking the initial state to be specifically the vacuum state $\rho = |\mathbf{0}\rangle \langle \mathbf{0}|$, the Lindbladian $\mathcal{L}^\dagger _0$ mixes rapidly, i.e.~in logarithmic time, with an upper bound \begin{equation}
         t_\textup{mix} \leq \frac{1}{2\Delta_0} \log\left[ \frac{1+\sqrt{3}}{4}\left(\coth\left(\frac{\beta}{2} \epsilon_\textup{min}\right) - \frac{1}{2}\right)\left(\coth\left(\frac{\beta}{2} \epsilon_\textup{min}\right) - 1\right) \cdot \frac{n}{\epsilon}\right]\,.
     \end{equation}
\end{prop}

\begin{remark}
    The initial state can be straightforwardly generalised to any convex combination of Gaussian states.
\end{remark}

\begin{proof}
    First, we need to recognise that when we start with a Gaussian state $\rho = |\mathbf{0}\rangle \langle \mathbf{0}|$ and evolve it with a quadratic Lindbladian, we will stay within the subspace of Gaussian states. These can be uniquely characterised by their vector of first moments $\mathbf{m} = \left\langle \left( \begin{matrix}
        \mathbf{x}\\\mathbf{p}
    \end{matrix}\right) \right\rangle = \Tr \left(\left( \begin{matrix}
        \mathbf{x}\\\mathbf{p}
    \end{matrix}\right) \rho \right)$ and their covariance matrices $\Gamma_{ij} = \frac{1}{2} \Tr(\{\Delta o_i,\Delta o_j\}\rho)$, where
    \begin{align}
        o_i= \begin{cases}
        x_i\ \text{for }1\leq i \leq n\\p_{i-n}\ \text{for }n+1\leq i \leq 2n
    \end{cases},
    \end{align}
    and then $\Delta o_i = o_i - \langle o_i \rangle = o_i - m_i$. Following \cite{koga2012dissipation}, we find that $\frac{\dd \mathbf{m}}{\dd t} = A \cdot \mathbf{m}$ and $\frac{\dd \Gamma}{\dd t} = A \Gamma + \Gamma A^T + D$, where for our Lindbladian we have
    \begin{align}
        A = -2 \left( \begin{matrix}
        q(h)^2 \sinh(h\beta/2) & 0 \\ 0 & q(h)^2 \sinh(h\beta/2)
    \end{matrix}\right)\quad\text{and}\quad D = 2 \left( \begin{matrix}
        q(h)^2 \cosh(h\beta/2) & 0 \\ 0 & q(h)^2 \cosh(h\beta/2)
    \end{matrix}\right).
    \end{align}
    Now, note that the initial vector of first moments is simply $\mathbf{m}(0) = \mathbf{0}$, while the initial covariance matrix is $\Gamma(0) = \frac{I}{2}$. Since all the terms in the equation and the initial condition commute with $I_2 \otimes h$, we can expect that so will $\Gamma(t)$, and hence find that $\Gamma(t) = I_2 \otimes \Gamma^{++}(t)$, where \begin{equation}
        \Gamma^{++}(t) = \left(\frac{I}{2} - \frac{1}{2}\coth(h\beta/2)\right) \cdot e^{-4q(h)^2 \sinh(h\beta/2) \cdot t} + \frac{1}{2} \coth(h\beta /2)\,,
    \end{equation} as well as $\mathbf{m}(t) = 0$. Further, we may observe that the Gibbs state has covariance matrix $\Gamma_{\sigma_\beta} = I_2 \otimes \frac{1}{2} \coth(h\beta /2) = \Gamma(\infty)$ and first moments $\mathbf{m}_{\sigma_\beta} = \mathbf{0} = \mathbf{m}(\infty)$, and so the evolution indeed converges to the Gibbs state. 

    Finally, we can use optimal trace norm bounds obtained in \cite{bittel2025optimalestimatestracedistance} which tells us that \begin{equation}
        \|\rho_1 - \rho_2 \|_{\Tr} \leq \frac{1+\sqrt{3}}{4} \max\{\|\Gamma_1\|, \|\Gamma_2\| \} \|\Gamma_1 - \Gamma_2\|_{\Tr} + \sqrt{\frac{ \min\{\|\Gamma_1\|, \|\Gamma_2\| \}}{8}} \|\mathbf{m}_1 - \mathbf{m}_2\|_2\,,
    \end{equation}
    from which we obtain \begin{align}
        \|\rho(t)-\sigma_\beta \|_{\Tr} &\leq \frac{1+\sqrt{3}}{4} \max\{\|\Gamma(t)\|, \|\Gamma_{\sigma_\beta}\| \} \|\Gamma(t) - \Gamma_{\sigma_\beta}\|_{\Tr}\\
        &\leq \frac{1+\sqrt{3}}{4}\left(\coth\left(\frac{\beta}{2} \epsilon_\textup{min}\right) - \frac{1}{2}\right)\left(\coth\left(\frac{\beta}{2} \epsilon_\textup{min}\right) - 1\right) \cdot n \cdot e^{-4\min_i q(\epsilon_i)^2 \cdot \sinh\left(\beta\epsilon_i /2\right) \cdot t}\\
        &\overset{\text{set}}{\leq } \epsilon \,,
    \end{align} which we can readily solve for $t$ and hence deduce that \begin{equation}
         t_\textup{mix} \leq \frac{1}{4\min_i q(\epsilon_i)^2 \cdot \sinh\left(\beta\epsilon_i /2\right)} \log\left[ \frac{1+\sqrt{3}}{4}\left(\coth\left(\frac{\beta}{2} \epsilon_\textup{min}\right) - \frac{1}{2}\right)\left(\coth\left(\frac{\beta}{2} \epsilon_\textup{min}\right) - 1\right) \cdot \frac{n}{\epsilon}\right]\,.
     \end{equation}
\end{proof}

\begin{remark}
    Since the underlying Hilbert space for bosons is infinite dimensional, we cannot map them to qubits without truncation. Simulating bosons without truncation would hence require a continuous-variable (CV) quantum computer, like a photonic quantum computer \cite{O_Brien_2009}. The state of algorithmic primitives for CV systems is not as well developed as for qubit systems, and for example the LCU used for the simulation of the Lindbladian has only recently been analysed for infinite-dimensional systems \cite{lu2025infinitedimensionalextensionlinearcombination}, as the creation and annihilation operators appearing in the algorithm are unbounded. Because of this, we do not yet fully understand the simulation techniques necessary for implementing these algorithmic bosonic Lindbladians.
\end{remark}


\section{Outlook}\label{sec:Outlook}

\paragraph*{Practical quantum advantages.} It is of great interest for the quantum computing community to find practical applications with large end-to-end speed-ups compared to state-of-the-art classical methods. Quantum Gibbs sampling has the potential to be one such application. However, before the era of large-scale fault-tolerant quantum computers, understanding its complexity relies on theoretical results about the mixing times and our work makes a step in this direction. We provide the first proof of efficiency of quantum Gibbs state preparation for weakly-interacting, non-commuting, qudit Hamiltonians, as well as the first rapid mixing result for any bosonic Lindbladian. Nevertheless, it is often said that for a firm quantum advantage, we would need to understand the behaviour of more strongly-correlated regimes, such as the intermediate-strength coupling regime of the Fermi-Hubbard model \cite{Qin2022hubbardmodel}. Hence, it would be of great importance to bound the mixing time of some physically relevant model throughout its whole parameter regime. Such a result would most likely require a plethora of novel proof techniques, tailored to the model at hand, which would also need to be specific enough to avoid any known complexity-theoretic hardness results.

Understanding mixing times of quantum Gibbs samplers experimentally on more near-term hardware is further obstructed by the need of complex algorithmic primitives, like block encodings and linear combinations of unitaries. As such, a recent strand of research focused on developing simplified algorithms, which approximately follow dynamics generated by exactly detailed-balanced Lindbladians\,---\,all while restricting themselves to only readily-implementable primitives, like Hamiltonian simulation via Trotterization \cite{Hahn:2025via,Hahn:2025eqt,Lloyd:2025cvp,ding2025endtoendefficientquantumthermal}.

Furthermore, one could probe the efficiency of quantum Gibbs samplers for larger systems using classical simulations based on tensor network techniques \cite{biamonte2017tensornutshell,mortier2024tensornetworks,zhan2025rapidquantumgroundstate}. The iTEBD and iDMRG algorithms can be even used to understand the mixing times of translationally invariant systems in the thermodynamic limit at least in 1D.\\


\paragraph*{Stability of rapid mixing within gapped Lindbladian phases.} We have covered the case of perturbed non-hopping fermions, but that leaves out many important examples of weakly interacting fermionic systems, like the Fermi-Hubbard model. The techniques currently at hand are falling short of proving general rapid mixing for their corresponding Gibbs samplers. Showing this result for free fermions required first diagonalising the fermions, which violated the locality of the system, and hence couldn't be extended to the perturbed system. The Lindblad operators in the unperturbed Lindbladian are also not short-range (only quasi-local), which prevents term-wise diagonalisation techniques akin to the qudit case presented.

Nevertheless, from \cite{smid2025polynomial,Tong_2025Fast} we already know that the perturbed Lindbladian for weakly interacting fermionic systems remains gapped. Moreover, denoting the Hamiltonian with interaction strength $s\lambda$ by $H(s) = H_0 + s\lambda \cdot V$, and $\mathcal{L}(s)$ its corresponding Lindbladian, then $\{\mathcal{L}(s):s \in [0,1]\}$ represents a sufficiently smooth\,---\,say, with a derivative $\frac{\dd \mathcal{L}(s)}{\dd s}$ bounded uniformly in $s$\,---\,path of quasi-local and detailed-balanced Lindbladians connecting the non-interacting Lindbladian $\mathcal{L}_0 = \mathcal{L}(0)$ to the interacting Lindbladian $\mathcal{L} = \mathcal{L}(1)$, along which the Lindbladians remain uniformly gapped with spectral gaps lower bounded by $\Delta$, which is independent of $s$ and $n$. This then immediately begs the intriguing question:

\begin{tcolorbox}[colback=yellow!15,colframe=yellow!15, boxrule=0pt, left=2pt, right=2pt, top=2pt, bottom=2pt]
\noindent Can we use this knowledge together with rapid mixing of $\mathcal{L}_0$ to prove rapid mixing of $\mathcal{L}$? That is, can we argue that rapid mixing is stable within a {\it topological phase} of Lindbladians?
\end{tcolorbox}

This line of thought can be further motivated by \cite{Cubitt_2015}, where it was shown that local observables evolved under a rapidly mixing Lindbladian are stable, i.e.~that a bounded perturbation of the Lindbladian terms results in a bounded difference of the observables, independent of the evolution time $t$ and the system size $n$. (For global observables, this difference could be polynomially large in $n$.) Intuitively, such a result together with the result on the gap remaining open, hence ensuring that the perturbation does not cross any topological obstacles, should indicate that the mixing time of the perturbed Lindbladian does not change by more than some constant factors, and thus remains rapid. We mark this as an important direction for future research, which would provide a more general method for showing rapid mixing of Lindbladians.\\


\acknowledgements

We thank Toby Cubitt, Daniel Stilck França, Angelo Lucia, and Cambyse Rouzé for helpful discussions.
All authors acknowledge support from the EPSRC Grant number EP/W032643/1. MB acknowledges funding by the European Research Council (ERC Grant Agreement No.~948139) and the Excellence Cluster Matter and Light for Quantum Computing (ML4Q).


\newpage
\bibliographystyle{alphaurl}
\bibliography{bibliography}


\appendix

\newpage
\section{Locality and strength of the Lindbladian}

\begin{lemma}\label{lemma:bounds on delta L}
    Consider a system of $n$ qudits of dimension $d$.
    Let $\mathcal{L}^0$ be the algorithmic Lindbladian associated to the separable Hamiltonian $H^0=\sum_ih_i$ and $\mathcal{L}$ that associated to $H=H^0+\lambda V$. 
    Assume that 
    \begin{enumerate}
        \item $V=\sum_{r\ge 1}\sum_{C\in C(r)}W_C$, where $C(r)$ is the set of balls of radius $r$ and $\max_{C\in C(r)}\|W_C\|\le K \e^{-\nu r}$.
        \item The jump operators for $\mathcal{L},\mathcal{L}^0$ are
        $A_{i,\nu}$ for $i\in [n]$ and $\nu\in \mathcal{J}$ indexing a basis of operators at site $i$.
    \end{enumerate}
    Denote by $\mathcal{L}_i=\sum_{\nu}\mathcal{L}_{i,\nu}$ the Lindbladian associated with jump operators at site $i$ and by $\mathcal{L}^{(r)}_i$ the Lindbladian where $H$ is replaced by $H_{B_r(i)}$, its truncation to a ball of radius $r$ around $i$.
    Then there are constants $C,\chi,C'\ge 0$ such that:
    \begin{align}
        &\| \mathcal{L}^{(r)}_i - \mathcal{L}^{(r-1)}_i\|_{\infty\to\infty}\le C\e^{-\chi r}\equiv \zeta(r)
        \,,\quad 
        \| \mathcal{L}_i - \mathcal{L}^0_i\|_{\infty\to\infty}\le
        C'|\lambda|
        \equiv
        \xi(\lambda)
        \,.
\end{align}
\end{lemma}

\begin{proof}    
Prop.~II.7 of \cite{smid2025polynomial} already contains the results on  locality. We review that here filling in some details.
Prop.~II.7 of \cite{smid2025polynomial} shows that 
\begin{align}
    &\|L_{i,\nu}^{(r)} - L_{i,\nu}^{(r-1)}\|
    \le 
    C_L \e^{-\mu_L r}
    \,,\quad 
    \mu_L=\mu
    \,,
    \quad 
    C_L
    =(1+\e^{\mu}) 
    J
    \int_{-\infty}^\infty |f(t)|   \e^{\mu v |t|}\ \dd t\,,
\end{align}
where constants $\mu,v,J$ are those of the Lieb-Robinson bound
\cite[Lemma 5]{Haah_2021} for exponentially decaying Hamiltonian interactions
\begin{equation}
    \left\| e^{iHt} A_{i,\nu} e^{-iHt} - e^{iH_{B_r(i)}t} A_{i,\nu}  e^{-iH_{B_r(i)}t}\right\| \leq \|A_{i,\nu}\| \min\left\{ 2, J e^{-\mu r} (e^{\mu v |t|}-1)\right\}\,.
\end{equation}
If we take the Gaussian filter function, $f(t)=\sqrt{\frac{2}{\pi\beta^2}}\e^{-\frac{2}{\beta^2}(t-i\frac{\beta}{4})^2}$, we have:
\begin{align}
    C_L
    &=
    (1+\e^{\mu}) 
    J
    c\,,\\
    c&=
    \int_{-\infty}^\infty    
    |f(t)|
    \e^{\mu v |t|}\ \dd t
    \\
    &=
    \sqrt{\pi}
    \e^{\frac{1}{8}((\beta\mu\nu)^2+1)}
    \left(\text{erf}
    \left(
    \frac{\beta\mu\nu}{2\sqrt{2}}
    \right)
     +1\right)
   \,.
\end{align}

Similarly, from Prop.~II.7 of \cite{smid2025polynomial} we have the following bound (note that in this reference the result is for $\tilde{G}_{i,\nu}=\sigma^{-1/4}G_{i,\nu}\sigma^{1/4}$, but it also applies to $G_{i,\nu}$ since the latter differs only by the presence of $1/\sinh(2\pi |t|/\beta)$ in the integrand, which can be bound like $1/\cosh(2\pi t/\beta)$ by $2\e^{-2\pi |t|/\beta}$):
\begin{align}
    \|G_{i,\nu} - G_{i,\nu}^{(r)}\|
    &\le 
    \frac{8}{\pi}\
    e^{1/8} \e^{-2\pi(\tilde c + r/v)/\beta} + 
    \frac{4 J}{\pi} e^{-\mu r+1/8} \left(e^{-2\pi(\tilde c + r/v)/\beta} -1\right) \\
    &\qquad- 
    \frac{8 c J e^{1/8} }{2\pi - \beta\mu v} \left(e^{-2\pi r /(v\beta) + \tilde c (\mu v - 2\pi /\beta)}-e^{-\mu r}\right)
    \\
    &\le 
    16
    \max\left(\frac{2}{\pi},\frac{ J}{\pi}, \frac{2 c  J}{|2\pi - \beta\mu v|} \right)\e^{\tilde{c}\mu v+1/8}
    \e^{-\min(\mu, \frac{2\pi}{\beta v}) r}
    \equiv 
    \tilde{C}_G \e^{-\mu_G r}\,,
\end{align}
where $\tilde{c}\mu v=
\log(2/(cJ))$
so that
\begin{align}
    \|G_{i,\nu}^{(r)} - G_{i,\nu}^{(r-1)}\|
    \le 
    \tilde{C}_G (1+\e^{\mu_G})\e^{-\mu_G r}
    \equiv 
    C_G \e^{-\mu_G r}
    \,.
\end{align}
Putting these together 
\begin{align}
    \| \mathcal{L}^{(r)}_i - \mathcal{L}^{(r-1)}_i\|_{\infty\to\infty}
    &\le
    \sum_{\nu\in\mathcal{J}}
    \Big(
    \| -i[G^{(r)}_{i,\nu}-G^{(r-1)}_{i,\nu},\cdot] \|_{\infty\to\infty}\\
    &\quad +
    \| (L^{(r)}_{i,\nu})^\dagger(\cdot)L^{(r)}_{i,\nu} - 
    (L^{(r-1)}_{i,\nu})^\dagger(\cdot)L^{(r-1)}_{i,\nu} \|_{\infty\to\infty}
    \\
    &\quad +
    \frac{1}{2}
    \| L^{(r)}_{i,\nu}(L^{(r)}_{i,\nu})^\dagger  (\cdot) - 
    L^{(r-1)}_{i,\nu}(L^{(r-1)}_{i,\nu})^\dagger  (\cdot) 
    \|_{\infty\to\infty}
    \\
    &\quad +
    \frac{1}{2}
    \| 
    (\cdot)L^{(r)}_{i,\nu}(L^{(r)}_{i,\nu})^\dagger   - 
    (\cdot) L^{(r-1)}_{i,\nu}(L^{(r-1)}_{i,\nu})^\dagger  \|_{\infty\to\infty}
    \Big)
    \\
    &\le 
    |\mathcal{J}|
    \left(
    2\| G^{(r)}_{i,\nu} - G^{(r-1)}_{i,\nu} \|
    +
    2\| L^{(r)}_{i,\nu} - L^{(r-1)}_{i,\nu} \|
    +
    2\| L^{(r)}_{i,\nu} - L^{(r-1)}_{i,\nu} \|
    \right)
    \\
    &\le
    2|\mathcal{J}|C_G\e^{-\mu_G r}
    +
    4|\mathcal{J}|C_L\e^{-\mu_L r}
    \\
    &\le 
    2|\mathcal{J}|
    \max(C_G,2C_L) \e^{-\min(\mu_G,\mu_L) r}
    \\
    &=
    2|\mathcal{J}|
    \max(C_G,2C_L) \e^{-\min(\mu,\frac{2\pi}{\beta v}) r}
    \equiv
    C\e^{-\chi r}
    \,.    
\end{align}
For bounding $\| \mathcal{L}_i-\mathcal{L}^0_i\|_{\infty\to\infty}$, similarly, we need to 
bound $\| L_{i,\nu} - L^0_{i,\nu} \|$ and $\| G_{i,\nu} - G^0_{i,\nu} \|$.
Note that $L^0_{i,\nu}$ and $G^0_{i,\nu}$ are supported at site $i$ only.
We have
\begin{align}
    \| L_{i,\nu} - L^0_{i,\nu} \|
    \le 
    \int_{-\infty}^\infty |f(t)|  
    \| \e^{iHt}A_{i,\nu}\e^{-iHt}-\e^{ih_it}A_{i,\nu}\e^{-ih_it}\|
    \ \dd t
\end{align}
and using Lemma A.1 of \cite{smid2025polynomial}:
\begin{align}
    \| L_{i,\nu} - L^0_{i,\nu} \|
    \le 
    |\lambda|
    \int_{-\infty}^\infty |f(t)|  
    |t|
    \max_{s\in[0,t]}
    \|[V,\e^{ih_is}A_{i,\nu}\e^{-ih_is}]\|
    \ \dd t
    \,.
\end{align}
Now by assumption, $V=\sum_{r\ge 1}\sum_{C\in C(r)}W_C$, where $C(r)$ is the set of balls of radius $r$ and $\max_{C\in C(r)}\|W_C\|\le K \e^{-\nu r}$. The number of terms $C\in C(r)$ that contain the site $i$ and thus have non-zero commutator equals $|B_i(r)|\le (2r+1)^D$, since each point in this ball can be seen as the centre of the support of a term $C$ that contains $i$. Then
\begin{align}
    \max_{s\in[0,t]}
    \|[V,\e^{ih_is}A_{i,\nu}\e^{-ih_is}]\|
    &\le 
    \max_{s\in[0,t]}
    \sum_{r\ge 1}(2r+1)^D
    \max_{C\in C(r)}
    \|[W_C,\e^{ish_i}A_{i,\nu}\e^{-ish_i}]\|\\
    &\le 
    2
    \sum_{r\ge 1}(2r+1)^D
    \max_{C\in C(r)}
    \|W_C \|
    \\
    &\le 
    2 K
    \sum_{r\ge 1}(2r+1)^D
    \e^{-\nu r}
\end{align}
and so
\begin{align}
    \| L_{i,\nu} - L^0_{i,\nu} \|
    &\le 
    2 K
    |\lambda|
    \int_{-\infty}^\infty |f(t)|  
    |t|
    \sum_{r\ge 1}(2r+1)^D
    \e^{-\nu r}
    \ \dd t
    \equiv 
    C_L'
    |\lambda| 
    \\
    C'_L
    &=
    2 K
    \int_{-\infty}^\infty |f(t)|  
    |t|
    \ \dd t
    \sum_{r\ge 1}(2r+1)^D
    \e^{-\nu r}\\
    &=
    2 K
    \int_{-\infty}^\infty 
    \sqrt{\frac{2}{\pi\beta^2}}
    \e^{1/4}
    \e^{-\frac{4}{\beta^2}t^2}
    |t|
    \ \dd t
    \sum_{r\ge 1}(2r+1)^D
    \e^{-\nu r}\\
    &=
    \frac{K }{\sqrt{2\pi}}
    \Phi_{D}(\e^{-\nu})
    \beta
    \,,\quad 
    \Phi_D(x)
    =
    \sum_{r\ge 1}(2r+1)^D
    x^{r}
\end{align}
Then, we look at $\| G_{i,\nu} - G^0_{i,\nu} \|$ and proceed similarly to Lemma III.8 in \cite{smid2025polynomial}. We have
\begin{align}
    G_{i,\nu} - G^0_{i,\nu}
    &=
    \int_{-\infty}^\infty 
    \dd t \,
    g(t) \Big( 
    e^{iHt} 
    \Big(
    L_{i,\nu}^\dagger L_{i,\nu}
    -
    (L_{i,\nu}^0)^\dagger L^0_{i,\nu}
    \Big) e^{-iHt}  \\
    &\qquad\qquad\qquad +
    e^{iHt} (L_{i,\nu}^0)^\dagger L_{i,\nu}^0 e^{-iHt}
    -
    e^{ih_it} (L_{i,\nu}^0)^\dagger L_{i,\nu}^0 e^{-ih_it}
    \Big)\,,
\end{align}
and so
\begin{align}
    \|G_{i,\nu}-G_{i,\nu}^0\|
    &\le 
    \int_{-\infty}^\infty 
    \dd t \, |g(t)| \Big( 
    \|e^{iHt} 
    \left(
    L_{i,\nu}^\dagger L_{i,\nu}
    -
    (L_{i,\nu}^0)^\dagger L^0_{i,\nu}
    \right)
    e^{-iHt}\|
    +
    \\
    &\qquad\qquad\qquad +
    \|e^{iHt} (L_{i,\nu}^0)^\dagger L_{i,\nu}^0 e^{-iHt}
    -
    e^{ih_it} (L_{i,\nu}^0)^\dagger L_{i,\nu}^0 e^{-ih_it}
    \|\,.
\end{align}
For the first term we have, using invariance of the operator norm by conjugation with $\e^{\beta/4 H}$, that
\begin{align}
    \int_{-\infty}^\infty 
    \dd t \, |g(t)| \Big( 
    \|e^{iHt} 
    \left(
    L_{i,\nu}^\dagger L_{i,\nu}
    -
    (L_{i,\nu}^0)^\dagger L^0_{i,\nu}
    \right)
    e^{-iHt}\|
    =
    \int_{-\infty}^\infty 
    \dd t \, |g(t+i\beta/4)| 
    \|
    L_{i,\nu}^\dagger L_{i,\nu}
    -
    (L_{i,\nu}^0)^\dagger L^0_{i,\nu}
    \|\,.
\end{align}
Now $g(t+i\beta/4)=-\frac{i}{\beta}\frac{1}{\cosh(2\pi/\beta t)}$ so that the integral converges and we can upper bound the first term by
\begin{align}
    2
    \int_{-\infty}^\infty 
    \dd t \, |g(t+i\beta/4)| 
    \|
    L_{i,\nu}
    -
    L^0_{i,\nu}
    \|
    \le 
    2
    C_L'|\lambda|
    \int_{-\infty}^\infty 
    \dd t \, |g(t+i\beta/4)| 
    \,.
\end{align}
Using again Lemma A.1 of \cite{smid2025polynomial} the second term is, denoting $B_{i,\nu}=(L_{i,\nu}^0)^\dagger L_{i,\nu}^0$:
\begin{align}
    \int_{-\infty}^\infty 
    \dd t \, |g(t)| 
    \|e^{iHt} B_{i,\nu} e^{-iHt}
    -
    e^{ih_it} B_{i,\nu}
    e^{-ih_it}
    \|
    &\le 
    |\lambda|
    \int_{-\infty}^\infty\dd t
    |g(t)t|
    \max_{s\in[0,t]}
    \|[V,\e^{ih_is}B_{i,\nu}
    \e^{-ih_is}]\|\\
    &\le 
    4K
    |\lambda| 
    \sum_{r\ge 1}(2r+1)^D
    \e^{-\nu r}
    \int_{-\infty}^\infty |g(t)t| 
    \,.
\end{align}
Note that the integral converges since the presence of $t$ in the integrand cancels the pole at $t=0$ of $g(t)$.
Finally,
\begin{align}
    \|G_{i,\nu}-G_{i,\nu}^0\|
    &\le 
    \frac{2}{\sqrt{2\pi}}
    K\Phi_D(\e^{-\nu})\beta
    |\lambda|
    \int_{-\infty}^\infty 
    \dd t \, |g(t+i\beta/4)| 
    +
    4K
    \Phi_{D}(\e^{-\nu})
    |\lambda|
    \int_{-\infty}^\infty\dd t 
    |g(t)t| 
    \\
    &=
    \frac{2}{\sqrt{2\pi}}
    K\Phi_D(\e^{-\nu})\beta
    |\lambda|
    \frac{1}{2}
    +
    4K
    \Phi_{D}(\e^{-\nu})
    |\lambda|
    \frac{\beta}{16}
    \\
    &=
    K\Phi_{D}(\e^{-\nu})\beta\left(
    \frac{1}{\sqrt{2\pi}}
    +\frac{1}{4}\right)|\lambda|
    \equiv 
    C_G' |\lambda|
    \,.
\end{align}
Then
\begin{align}
    \|\mathcal{L}_i-\mathcal{L}_i^0\|_{\infty\to\infty}
    \le 
    \sum_{\nu\in \mathcal{J}}
    2\| G_{i,\nu} - G^0_{i,\nu} \|
    +4\| L_{i,\nu} - L^0_{i,\nu} \|
    \le 
    |\mathcal{J}|
    (2C_G'+4C_L')|\lambda|\equiv C'|\lambda|
    \,.
\end{align}
\end{proof}

\newpage
\section{Lieb-Robinson bounds and localisation for odd fermionic operators}
\label{sec: Localising odd fermionic operators}

For our proofs and the explicit calculations of Section \ref{sec: Evaluating explicit parameters for FH model}, we require localising odd fermionic operators, e.g.~bounds of the form \begin{align}
    \left\|c_i(t) - c^{(r)}_i(t)\right\| \leq J e^{-\mu r} (e^{\mu v t}-1)\,,\label{eqn: Lieb-Robinson we want to show}
\end{align} where $c_i^{(r)}(t)$ represents a truncation of $c_i(t)$ to $B_r(i)$. However, previous literature studying explicit Lieb-Robinson bounds for fermionic systems, like the Fermi-Hubbard model \cite{Wang2020}, provides only bounds on the anticommutators of the form $\|\{c_i(t),c_j\}\| \leq \epsilon(d_{ij},t)$. To the best of our knowledge, subsequent rigorous treatments connecting bounds on commutators/anticommutators with those on locality are restricted to only even fermionic operators \cite{nachtergaele2018lieb}. In this section, we expand upon this work by showing how to localise odd fermionic operators, and we provide an explicit bound of the form \eqref{eqn: Lieb-Robinson we want to show} for the Fermi-Hubbard model on a $D$-dimensional cubic lattice.

In order to map a fermionic operator to its local approximation, we require the notion of a conditional expectation. As was shown in \cite{araki2003} and discussed in \cite{nachtergaele2018lieb}, if we require a treatment of both odd and even operators, then for any subset $X \subseteq \Lambda$ of the lattice, there exists a unique conditional expectation $\mathbb{F}^\Lambda_X:\mathcal{A}_\Lambda \to \mathcal{A}_X$, where $\mathcal{A}_B$ is the CAR algebra generated by the creation and annihilation operators on the set $B$, which leaves the tracial state invariant, i.e.~satisfies \begin{equation}\Tr\left( \mathbb{F}_X^\Lambda (A B) \right) = \Tr\left( \mathbb{F}_X^\Lambda (A) B \right)\end{equation} for all $A \in \mathcal{A}_\Lambda$ and $B \in \mathcal{A}_X$. This conditional expectation can then be expressed in its Krauss form as \begin{align}
    \mathbb{F}^\Lambda_X(A) = \frac{1}{4^{|\Lambda\backslash X|}} \sum\limits_{\alpha \in I_{\Lambda\backslash X}} \tilde{u}(\alpha)^\dagger A \tilde{u}(\alpha)\,.
\end{align}
Here, for each site $i\in \Lambda$, we define 4 unitary generators of the algebra \begin{align}
    u_i^{(0)} = I,\quad u_i^{(1)} = c_i^\dagger + c_i,\quad u_i^{(2)}=c_i^\dagger - c_i,\quad u_i^{(3)}=I-2c_i^\dagger c_i\,.
\end{align} Then for any subset $X \subseteq \Lambda$, we consider a multi-index $\alpha \in I_{\Lambda\backslash X} = \{0,1,2,3\}^{|\Lambda\backslash X|}$ and a fixed ordering of the sites $\Lambda\backslash X = \{i_1,\dots,i_{|\Lambda\backslash X|}\}$, and define the unitary operators $u(\alpha) \in \mathcal{A}_{\Lambda\backslash X}$ by \begin{align}
    u(\alpha) = u_{i_1}^{(\alpha(i_1))} \cdots u_{i_{|\Lambda\backslash X|}}^{(\alpha(i_{|\Lambda\backslash X|}))}\,.
\end{align} Finally, depending on the parity of $u(\alpha)$, we define \begin{align}
    \tilde{u}(\alpha) = \begin{cases}
        u(\alpha) & \text{if } \alpha \in I^+_{\Lambda\backslash X}\,,\\
        (-1)^{N_X} \cdot u(\alpha) & \text{if } \alpha \in I^-_{\Lambda\backslash X}\,,
    \end{cases}
\end{align} where $N_X = \sum\limits_{i\in X} c_i^\dagger c_i$ is the number operator on the subset $X$, and $I^\pm_{\Lambda\backslash X}$ denotes the set of multi-indices $\alpha$ corresponding to even or odd operators $u(\alpha)$ respectively.

With this setup, we can prove the following lemma:
\begin{lemma}\label{lemma: localising odd fermionic operators}
    Let $A\in \mathcal{A}_{\Lambda}^-$ be an odd fermionic operator, $X \subseteq \Lambda$ a subset of the lattice, and $\epsilon > 0$. If we have $\|[A,B]\| \leq \epsilon \|B\|$ for all even operators $B\in \mathcal{A}^+_{\Lambda \backslash X}$ and $\|\{A,B\}\| \leq \epsilon \|B\|$ for all odd operators $B\in \mathcal{A}^-_{\Lambda \backslash X}$, then there exists $A' = \mathbb{F}_X^\Lambda(A) \in \mathcal{A}^-_X$ such that $\|A-A'\| \leq \epsilon$.
\end{lemma}
\begin{proof}
    Consider \begin{align}
        A-\mathbb{F}_X^\Lambda(A) &= \frac{1}{4^{|\Lambda\backslash X|}} \sum\limits_{\alpha \in I_{\Lambda\backslash X}} \left(A- \tilde{u}(\alpha)^\dagger A \tilde{u}(\alpha)\right)\\
        &= \frac{1}{4^{|\Lambda\backslash X|}} \sum\limits_{\alpha \in I^+_{\Lambda\backslash X}} \left(A- u(\alpha)^\dagger A u(\alpha)\right) + \frac{1}{4^{|\Lambda\backslash X|}} \sum\limits_{\alpha \in I^-_{\Lambda\backslash X}} \left(A- u(\alpha)^\dagger (-1)^{N_X} A (-1)^{N_X} u(\alpha)\right)\\
        &= \frac{1}{4^{|\Lambda\backslash X|}} \sum\limits_{\alpha \in I^+_{\Lambda\backslash X}} u(\alpha)^\dagger [u(\alpha),A] + \frac{1}{4^{|\Lambda\backslash X|}} \sum\limits_{\alpha \in I^-_{\Lambda\backslash X}} \left(A- u(\alpha)^\dagger (-1)^{2\cdot N_{\Lambda\backslash X}} (-1)^{N_X} A (-1)^{N_X} (-1)^{2\cdot N_{\Lambda\backslash X}} u(\alpha)\right)\\
        &= \frac{1}{4^{|\Lambda\backslash X|}} \sum\limits_{\alpha \in I^+_{\Lambda\backslash X}} u(\alpha)^\dagger [u(\alpha),A] + \frac{1}{4^{|\Lambda\backslash X|}} \sum\limits_{\alpha \in I^-_{\Lambda\backslash X}} \left(A- u(\alpha)^\dagger (-1)^{\cdot N_{\Lambda\backslash X}} (-1)^{N_\Lambda} A (-1)^{N_\Lambda} (-1)^{\cdot N_{\Lambda\backslash X}} u(\alpha)\right)\\
        &= \frac{1}{4^{|\Lambda\backslash X|}} \sum\limits_{\alpha \in I^+_{\Lambda\backslash X}} u(\alpha)^\dagger [u(\alpha),A] + \frac{1}{4^{|\Lambda\backslash X|}} \sum\limits_{\alpha \in I^-_{\Lambda\backslash X}} \left(A+ u(\alpha)^\dagger (-1)^{\cdot N_{\Lambda\backslash X}} A (-1)^{\cdot N_{\Lambda\backslash X}} u(\alpha)\right)\\
        &= \frac{1}{4^{|\Lambda\backslash X|}} \sum\limits_{\alpha \in I^+_{\Lambda\backslash X}} u(\alpha)^\dagger [u(\alpha),A] + \frac{1}{4^{|\Lambda\backslash X|}} \sum\limits_{\alpha \in I^-_{\Lambda\backslash X}} u(\alpha)^\dagger (-1)^{\cdot N_{\Lambda\backslash X}} \{ (-1)^{\cdot N_{\Lambda\backslash X}} u(\alpha),A\}\,,
    \end{align} where the second-to-last equality follows as $A\in \mathcal{A}_\Lambda ^-$ is odd.
    Next, by taking the norm, and using triangle inequality together with the unitarity of $u(\alpha)$'s, we obtain \begin{align}
        \|A-\mathbb{F}_X^\Lambda(A)\| &\leq \frac{1}{4^{|\Lambda\backslash X|}} \sum\limits_{\alpha \in I^+_{\Lambda\backslash X}} \| [u(\alpha),A]\| + \frac{1}{4^{|\Lambda\backslash X|}} \sum\limits_{\alpha \in I^-_{\Lambda\backslash X}} \| \{ (-1)^{\cdot N_{\Lambda\backslash X}} u(\alpha),A\}\|\\
        &\leq \epsilon\,,
    \end{align} where the final inequality follows from the assumptions of the lemma.
\end{proof}

We may further decompose this upper bound to only involve the anticommutators with single-site creation and annihilation operators supported outside of the truncation subset. First observe that $
    \{BC,A\} = BCA + ABC = B \{C,A\} + [A,B]C\,,
$ and hence get that \begin{equation}
    \| \{ (-1)^{\cdot N_{\Lambda\backslash X}} u(\alpha),A\}\| \leq \|\{u(\alpha),A\}\| + \|[(-1)^{\cdot N_{\Lambda\backslash X}}, A]\|\,,
\end{equation} so that \begin{align}
        \|A-\mathbb{F}_X^\Lambda(A)\| &\leq \frac{1}{4^{|\Lambda\backslash X|}} \sum\limits_{\alpha \in I^+_{\Lambda\backslash X}} \| [u(\alpha),A]\| + \frac{1}{4^{|\Lambda\backslash X|}} \sum\limits_{\alpha \in I^-_{\Lambda\backslash X}} \| \{ u(\alpha),A\}\| + \frac{1}{2}\|[(-1)^{\cdot N_{\Lambda\backslash X}}, A]\|\,.
    \end{align}
Next, by writing $(-1)^{\cdot N_{\Lambda\backslash X}} = e^{i\pi N_{\Lambda\backslash X}}$ and using the formula $
    [A,e^B] = \int\limits_0^1 e^{(1-s)B} [A,B]e^{sB}\ \dd s $, which holds for any operators $A$ and $B$, we can find that \begin{equation}
        \|[(-1)^{\cdot N_{\Lambda\backslash X}}, A]\| \leq \pi \cdot \|[A,N_{\Lambda\backslash X}]\| \leq \pi \cdot \sum_{i \in \Lambda\backslash X} \|[A,c_i^\dagger c_i]\| \leq \pi \cdot \sum_{i \in \Lambda\backslash X} \left(\|\{A,c_i^\dagger\}\| +  \|\{A,c_i\}\|\right)\,.
    \end{equation}
Regarding the $\|[A,u(\alpha)]_\pm\|$ terms, observe that for each $u(\alpha) = u_{i_1}^{(\alpha(i_1))} \cdots u_{i_{|\Lambda\backslash X|}}^{(\alpha(i_{|\Lambda\backslash X|}))} \in \mathcal{A^\pm}_{\Lambda\backslash X}$, we can decompose \begin{equation}
    \|[A,u(\alpha)]_\mp\| \leq \sum_{i \in \Lambda\backslash X} \|[A,u_i^{(\alpha(i))}]_{\overline{\sigma}(\alpha(i)))}\|\,,
\end{equation} where now each $u_i$ term is supported only at the site $i$ and the commutator/anticommutator is taken appropriately according to its parity. Finally, we have that \begin{equation}
    \|[A,u_i^{(0)}]\| = 0,\quad \|\{A,u_i^{(1,2)}\}\| \leq \|\{A,c_i^\dagger\}|| + \|\{A,c_i\}||,\quad \|[A,u_i^{(3)}]\| \leq 2\|\{A,c_i^\dagger\}|| + 2\|\{A,c_i\}||\,,
\end{equation} and so we can get the bound \begin{align}
    \|A-\mathbb{F}_X^\Lambda(A)\| &\leq \left(2+\frac{\pi}{2}\right) \cdot\sum_{i \in \Lambda\backslash X} \left(\|\{A,c_i^\dagger\}\| +  \|\{A,c_i\}\|\right)\,.
\end{align}

To conclude, we wish to obtain an explicit bound of the form \eqref{eqn: Lieb-Robinson we want to show} for the Fermi-Hubbard model on a $D$-dimensional cubic lattice, given by \begin{equation}
    H_\text{FH} = -t\sum_{\langle i,j\rangle, \sigma} \left( c^\dagger _{i,\sigma} c_{j,\sigma} + c^\dagger_{j,\sigma}c_{i,\sigma} \right) + U\sum_{i=1}^n N_{i,\uparrow} N_{i,\downarrow}\,.
\end{equation} Following \cite[Theorem 1]{hastings2010localityquantumsystems}, the only terms appearing in the Hamiltonian have diameter $0$ (the on-site interaction between different spins) and $1$ (the hopping of fermions between nearest neighbours). Hence, whenever we have \begin{align}
    |U| + 8D|t|e^\mu \leq s\,,
\end{align} we find that \begin{align}
    \| [c_i^{(\dagger)}(\tau),B]\|\leq 2e^{-\mu r} (e^{2s|\tau|}-1) \|B\|
\end{align} for any even operator $B \in \mathcal{A}^+_{\Lambda\backslash B_r(i)}$, and \begin{align}
    \| \{c_i^{(\dagger)}(\tau),B\}\|\leq 2e^{-\mu r} (e^{2s|\tau|}-1) \|B\|
\end{align} for any odd operator $B \in \mathcal{A}^-_{\Lambda\backslash B_r(i)}$. Using Lemma \ref{lemma: localising odd fermionic operators}, we then find the following bound on locality of creation/annihilation operators: \begin{align}
    \left\|c_i^{(\dagger)}(\tau) - c^{(\dagger)(r)}_i(\tau)\right\| \leq 2e^{-\mu r} (e^{2(|U| + 8D|t|e^\mu)|\tau|}-1)\,. \label{eqn: explicit bound on locality in FH model}
\end{align}
Note that this bound holds for any $\mu>0$. This degree of freedom comes from the fact that we are working with a short-range model, not a quasi-local one, and as such, its Lieb-Robinson bounds are actually decaying super-exponentially. However, we are only using the exponentially decaying bound, and so have the freedom to choose the coefficient in the exponent.

\newpage
\section{Strongly-interacting Fermi-Hubbard model:~eigenvectors $F_i^{(15,16)}$}\label{appendix: eigenvectors}

\rotatebox{90}{%
\parbox{1.5\textwidth}{
\tiny{\begin{align}
    F_i^{(15)} &\propto \frac{2 \left(2 f_+^2 f_-^4+\left(f_+^4+\left(\sqrt{4 f_-^4+4 \left(f_+^2-3\right) f_-^2+f_+^4-2 f_+^2+9}-5\right) f_+^2+6\right) f_-^2-2 f_+^4+2 \sqrt{4 f_-^4+4 \left(f_+^2-3\right) f_-^2+f_+^4-2 f_+^2+9} f_+^2+7 f_+^2-3 \sqrt{4 f_-^4+4 \left(f_+^2-3\right) f_-^2+f_+^4-2 f_+^2+9}-9\right)}{\left(f_-^2+f_+^2-2\right) \left(2 f_-^2+f_+^2+\sqrt{4 f_-^4+4 \left(f_+^2-3\right) f_-^2+f_+^4-2 f_+^2+9}+3\right) \left(2 f_-^4+\left(f_+^2+\sqrt{4 f_-^4+4 \left(f_+^2-3\right) f_-^2+f_+^4-2 f_+^2+9}-1\right) f_-^2-f_+^2+\sqrt{4 f_-^4+4 \left(f_+^2-3\right) f_-^2+f_+^4-2 f_+^2+9}+3\right)}\cdot I\\&-\frac{2 f_+^2 f_-^4+\left(f_+^4+\left(\sqrt{4 f_-^4+4 \left(f_+^2-3\right) f_-^2+f_+^4-2 f_+^2+9}-5\right) f_+^2+6\right) f_-^2-2 f_+^4+2 \sqrt{4 f_-^4+4 \left(f_+^2-3\right) f_-^2+f_+^4-2 f_+^2+9} f_+^2+7 f_+^2-3 \sqrt{4 f_-^4+4 \left(f_+^2-3\right) f_-^2+f_+^4-2 f_+^2+9}-9}{2 \left(f_-^2+f_+^2-2\right) \left(2 f_-^4+\left(f_+^2+\sqrt{4 f_-^4+4 \left(f_+^2-3\right) f_-^2+f_+^4-2 f_+^2+9}-1\right) f_-^2-f_+^2+\sqrt{4 f_-^4+4 \left(f_+^2-3\right) f_-^2+f_+^4-2 f_+^2+9}+3\right)}\cdot N_{i,\downarrow}\\&-\frac{2 f_+^2 f_-^4+\left(f_+^4+\left(\sqrt{4 f_-^4+4 \left(f_+^2-3\right) f_-^2+f_+^4-2 f_+^2+9}-5\right) f_+^2+6\right) f_-^2-2 f_+^4+2 \sqrt{4 f_-^4+4 \left(f_+^2-3\right) f_-^2+f_+^4-2 f_+^2+9} f_+^2+7 f_+^2-3 \sqrt{4 f_-^4+4 \left(f_+^2-3\right) f_-^2+f_+^4-2 f_+^2+9}-9}{2 \left(f_-^2+f_+^2-2\right) \left(2 f_-^4+\left(f_+^2+\sqrt{4 f_-^4+4 \left(f_+^2-3\right) f_-^2+f_+^4-2 f_+^2+9}-1\right) f_-^2-f_+^2+\sqrt{4 f_-^4+4 \left(f_+^2-3\right) f_-^2+f_+^4-2 f_+^2+9}+3\right)}\cdot N_{i,\uparrow}\\&+N_{i,\downarrow} N_{i,\uparrow}
\end{align}
\bigskip\bigskip
\begin{align}
    F_i^{(16)} &\propto -\frac{2 \left(-2 f_+^2 f_-^4+\left(-f_+^4+\left(\sqrt{4 f_-^4+4 \left(f_+^2-3\right) f_-^2+f_+^4-2 f_+^2+9}+5\right) f_+^2-6\right) f_-^2+2 f_+^4+2 \sqrt{4 f_-^4+4 \left(f_+^2-3\right) f_-^2+f_+^4-2 f_+^2+9} f_+^2-7 f_+^2-3 \sqrt{4 f_-^4+4 \left(f_+^2-3\right) f_-^2+f_+^4-2 f_+^2+9}+9\right)}{\left(f_-^2+f_+^2-2\right) \left(-2 f_-^2-f_+^2+\sqrt{4 f_-^4+4 \left(f_+^2-3\right) f_-^2+f_+^4-2 f_+^2+9}-3\right) \left(-2 f_-^4+\left(-f_+^2+\sqrt{4 f_-^4+4 \left(f_+^2-3\right) f_-^2+f_+^4-2 f_+^2+9}+1\right) f_-^2+f_+^2+\sqrt{4 f_-^4+4 \left(f_+^2-3\right) f_-^2+f_+^4-2 f_+^2+9}-3\right)}\cdot I\\
    &-\frac{-2 f_+^2 f_-^4+\left(-f_+^4+\left(\sqrt{4 f_-^4+4 \left(f_+^2-3\right) f_-^2+f_+^4-2 f_+^2+9}+5\right) f_+^2-6\right) f_-^2+2 f_+^4+2 \sqrt{4 f_-^4+4 \left(f_+^2-3\right) f_-^2+f_+^4-2 f_+^2+9} f_+^2-7 f_+^2-3 \sqrt{4 f_-^4+4 \left(f_+^2-3\right) f_-^2+f_+^4-2 f_+^2+9}+9}{2 \left(f_-^2+f_+^2-2\right) \left(-2 f_-^4+\left(-f_+^2+\sqrt{4 f_-^4+4 \left(f_+^2-3\right) f_-^2+f_+^4-2 f_+^2+9}+1\right) f_-^2+f_+^2+\sqrt{4 f_-^4+4 \left(f_+^2-3\right) f_-^2+f_+^4-2 f_+^2+9}-3\right)}\cdot N_{i,\downarrow}\\
    &-\frac{-2 f_+^2 f_-^4+\left(-f_+^4+\left(\sqrt{4 f_-^4+4 \left(f_+^2-3\right) f_-^2+f_+^4-2 f_+^2+9}+5\right) f_+^2-6\right) f_-^2+2 f_+^4+2 \sqrt{4 f_-^4+4 \left(f_+^2-3\right) f_-^2+f_+^4-2 f_+^2+9} f_+^2-7 f_+^2-3 \sqrt{4 f_-^4+4 \left(f_+^2-3\right) f_-^2+f_+^4-2 f_+^2+9}+9}{2 \left(f_-^2+f_+^2-2\right) \left(-2 f_-^4+\left(-f_+^2+\sqrt{4 f_-^4+4 \left(f_+^2-3\right) f_-^2+f_+^4-2 f_+^2+9}+1\right) f_-^2+f_+^2+\sqrt{4 f_-^4+4 \left(f_+^2-3\right) f_-^2+f_+^4-2 f_+^2+9}-3\right)}\cdot N_{i,\uparrow}\\
    &+N_{i,\downarrow} N_{i,\uparrow}
\end{align}

}}}%

\end{document}